\documentclass[a4paper,twocolumn,11pt,accepted=2025-07-22]{quantumarticle}
\pdfoutput=1

\usepackage{amsmath}

\usepackage{amssymb}
\usepackage{amsthm}
\usepackage{mathrsfs}
\usepackage{bm, dsfont}
\usepackage[usenames,dvipsnames]{color}
\usepackage{colortbl}
\usepackage[table]{xcolor}
\usepackage{enumitem}

\usepackage[colorlinks=true,linkcolor=blue,citecolor=blue,urlcolor=blue]{hyperref}
\usepackage{cleveref}
\usepackage{hypcap}
\usepackage{verbatim, float}
\usepackage{psfrag}
\usepackage[normalem]{ulem}
\usepackage{physics}		

\usepackage[demo]{graphicx}

\usepackage[caption=false]{subfig}

\newcommand{\id}{\mathds{1}}

\newcommand{\ve}{\varepsilon}

\newcommand{\cE}{\mathcal{E}}

\newcommand{\be}{\begin{equation}}
\newcommand{\ee}{\end{equation}}

\newcommand{\nc}{\textup{\,\o\,}}

\newcommand{\fe}{\text{f}_\varepsilon}

\let\origtau\tau 
\renewcommand{\tau}{\scalebox{1.44}{$\origtau$}}

\newtheorem{theorem}{Theorem}
\newtheorem{definition}{Definition}

\newtheorem{lemma}[theorem]{Lemma}

\newcommand{\name}{\text{TBSM}}

\renewcommand{\ve}{\varepsilon}

\newcommand{\unnumberedfootnote}[1]{%
  \begingroup
  \renewcommand{\thefootnote}{}
  \footnote{#1}%
  \endgroup
}

\begin{document}

\title{Noise-robust proofs of quantum network nonlocality}

\author{Sadra Boreiri$^\dagger$}
\affiliation{Department of Applied Physics, University of Geneva, Switzerland}
\author{Bora Ulu$^\dagger$}
\affiliation{Department of Applied Physics, University of Geneva, Switzerland}
\author{Nicolas Brunner}
\affiliation{Department of Applied Physics, University of Geneva, Switzerland}
\author{Pavel Sekatski}
\affiliation{Department of Applied Physics, University of Geneva, Switzerland}
\unnumberedfootnote{$^\dagger$ These authors contributed equally to this work}

\begin{abstract}
Quantum networks allow for novel forms of quantum nonlocality. By exploiting the combination of entangled states and entangled measurements, strong nonlocal correlations can be generated across the entire network. So far, all proofs of this effect are essentially restricted to the idealized case of pure entangled states and projective local measurements. Here we present noise-robust proofs of network quantum nonlocality, for a class of quantum distributions on the triangle network that are based on entangled states and entangled measurements. The key ingredient is a result of approximate rigidity for local distributions that satisfy the so-called ``parity token counting'' property with high probability. Our methods can be applied to any type of noise. As illustrative examples, we consider quantum distributions obtained with imperfect sources and obtain a noise robustness up to $\sim 80\%$ for dephasing noise and up to $\sim 0.5\%$ for white noise. Additionally, we prove that all distributions in the vicinity of some ideal quantum distribution are nonlocal, with a bound on the total-variation distance $\sim 0.25\%$. 
Our work opens interesting perspectives towards the practical implementation of quantum network nonlocality.
 \end{abstract}

\maketitle

\section{Introduction}

A growing interest has recently been devoted to the question of quantum nonlocality in networks, see e.g. \cite{Tavakoli_Review} for a review. A general framework has been developed to investigate nonlocal correlations in networks featuring independent sources \cite{Branciard2010,Branciard_2012,Fritz_2012,Rosset2016,wolfe2019inflation,renou2019limits,Aberg2020,wolfe2021inflation,Contreras20,gisin2020constraints,Ming_Luo_1,Ming_Luo_2}. A central motivation comes from the fact that the network setting allows for novel forms of nonlocality compared to the more standard Bell test, where a single source distributes a physical resource (e.g. an entangled quantum state) to the parties. 

\begin{figure}[b!]
    \centering
    \includegraphics[width=0.8\columnwidth]{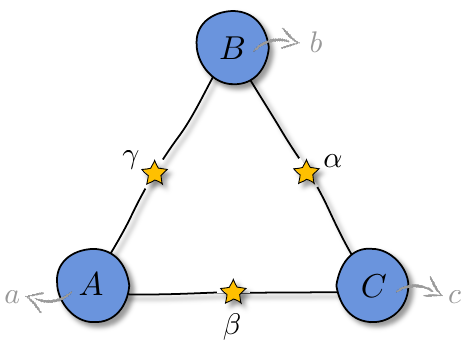}
    \caption{The triangle network features three distant parties, connected pairwise by three independent sources distributing physical systems.}
\end{figure}

Notably, it is possible to observe quantum nonlocality with parties performing a fixed measurement \cite{Fritz_2012,Branciard_2012}; see also \cite{Polino2023} for a recent experiment. That is, in each round of the experiment, each party provides an output, but they receive no input. This in contrast to the standard Bell test where the use of measurement inputs is essential.

Another key aspect of the network setting is the possibility to exploit joint entangled measurements (where one or more eigenstates are non-separable) \cite{gisin2019entanglement}. Similarly, as in quantum teleportation, the combination of entangled states and entangled joint measurements allows for strong correlations to be distributed across the entire network. Notably, this leads to novel forms of quantum nonlocal correlations that are genuine to the network structure \cite{Renou_2019}, in the sense that they crucially rely on the use of entangled joint measurements \cite{Supic2022,Sekatski2023}. This is possible for a simple triangle network, as in Fig. 1.

While these works represent significant progress in our understanding of quantum nonlocality, they suffer from one major limitation. Namely, these proofs of genuine quantum network nonlocality are derived in an ideal noiseless setting, where the sources distribute ideal (pure) quantum states and nodes perform ideal (projective) quantum measurements. Formally, these proofs are based on the concept of ``token counting'' rigidity \cite{Renou_2019,renou2022network,Boreiri2023}, which so far only applies to the ideal (noiseless) case. Hence, these nonlocality proofs no longer apply when any realistic noise level is considered. This represents a major hurdle towards any practical implementation or application of these ideas, as any experiment features a certain amount of noise originating from unavoidable technical imperfections. Moreover, numerical analysis based on neural networks indicates noise-robust nonlocality \cite{krivachy_neural_2020}.

One can draw a parallel with the situation of nonlocality in the mid-sixties. Indeed, John Bell's initial proof of nonlocality applied only to the idealized case of a pure entangled state and projective local measurements \cite{Bell}. A key step was then made by Clauser, Horne, Shimony, and Holt who derived a noise-robust proof of quantum nonlocality via their famous CHSH Bell inequality \cite{CHSH}, which enabled the first Bell experiments \cite{CHSH_ex_2,CHSH_ex_3}. This is still today the most commonly used test of nonlocality \cite{Hensen_2016, shalm2015strong,Zeilinger_Bell}, and the basis of many applications within the device-independent approach (see e.g. \cite{review}).

In this work we address the question of noise-robust
proofs of quantum network nonlocality. We derive methods for this problem and apply them to several examples of quantum nonlocal distributions that are based on the concept of ``parity token counting'' (PTC) \cite{Boreiri2023}. Our main contribution is a result of approximate rigidity for local distributions that satisfy the PTC property with a high (but non-unit) probability. Importantly, this method can be applied to any type of noise. We present illustrative examples, considering
quantum distributions obtained from imperfect sources, featuring either dephasing or white noise. Additionally, we show that all distributions in the vicinity of some ideal quantum distribution are nonlocal, with a bound on the total-variation
distance, which can be viewed as a nonlinear Bell inequality for the triangle network. Our work opens interesting perspectives towards the practical implementation of Bell nonlocality in quantum networks.

\section{Summary of results}

To start, we first provide a summary of the main contributions of this work. 

\begin{enumerate}
\item We present a general family of quantum distributions on the triangle network (without inputs) and prove their nonlocality in the ideal (noiseless) case, see Section \ref{sec: TBSM}. These distributions include as a special case the distributions of Renou et al. \cite{Renou_2019}, from now on referred to as RGB4. The nonlocality proof is based on the  ``parity token counting'' (PTC) rigidity property~\cite{Boreiri2023}. In words, it states that on the triangle network with binary outcomes, any local distribution for which the sum of the outputs is always odd (1 or 3) has a unique corresponding local model.
\item We derive noise-robust nonlocality proofs for these quantum distributions, see Section~\ref{sec: WN}. Our methods can be applied to arbitrary type of noise. As an illustration, we demonstrate the nonlocality of quantum distributions where up to $\sim 80\%$ dephasing noise, or $\sim 0.55$\% white noise, is added at each source. We also investigated other noise models, such as no-click and photon loss; see Appendix~\ref{app: other noise}. Moreover, in Section~\ref{sec: TVD} we characterize the distributions in the vicinity of a given nonlocal distribution.  For a specific quantum distribution, we can show that any distribution in a $\sim 0.25$\% total-variation distance ball around it is nonlocal. This can be viewed as a Bell inequality, in the sense that it guarantees that a full-measure region of the correlation set (the ball around the quantum distribution) is nonlocal.
 \item The key ingredient which allows us to prove noise-robust nonlocality, is a result on the approximate PTC rigidity, discussed in Section~\ref{sec: approximate PTC}. It gives a quantitative characterization of all local models leading to distributions for which the PTC condition holds approximately. That is, the distributions for which the sum of the outputs is odd with high probability.  With the help of this result, we also derive an analytical outer relaxation of the set of binary output correlations compatible with a triangle-local model.
\end{enumerate}

 Below we start with Section~\ref{sec: backgrounds} where we introduce more formally the context, the notations, and review the concept of parity token counting (PTC) distributions.

\section{Background}

\label{sec: backgrounds}

 To start we recall some background notions and fix the notations. Consider a triangle network as in Fig. 1, where three independent sources (labeled $\alpha$, $\beta$,  and $\gamma$) distribute physical systems to three distant parties (nodes $A$, $B$, and $C$). Each party thus receives two physical subsystems (from two different sources), and produces a measurement outcome, denoted by the classical variables $a$, $b$, and $c$ (of finite cardinality $d$). The experiment results in the joint probability distribution $P(a,b,c)$. 
 
 We aim to understand which distributions are compatible with the triangle causal network depicted in Fig. 1. The set of possible distributions depends on the underlying physical theory. We are interested here in distributions that can be realized in classical physics (where the sources produce correlated classical variables) and in quantum theory (where the sources can produce entangled states). Our focus will be on quantum distributions that cannot be realized classically, hence demonstrating network nonlocality.

 First, a distribution $P_L(a,b,c)$ is said to be triangle-local if it admits a classical model, i.e. it can be decomposed as 
\be\label{eq: P local}
P_L(a,b,c) = \mathds{E}\Big( P_A(a| \beta, \gamma)  P_B(b|\gamma,\alpha) P_C(c| \alpha, \beta) \Big),
\ee
where $\mathds{E}$ is the expected value over the independently distributed classical variables $ \alpha, \beta, \gamma$ \cite{rosset2017universal}. Furthermore, the conditional probabilities $P_X(x|\xi,\xi')$ can be considered deterministic without loss of generality, and represented by the response functions $x(\xi,\xi')$. A distribution that does not admit a decomposition of the form \eqref{eq: P local} is termed triangle-nonlocal. For brevity, we refer to triangle-(non)local correlations simply as (non)local throughout the text.

Second, we say that a distribution $P_Q(a,b,c)$ is triangle-quantum if it can be written as
\be\begin{split}
P_Q(a,b,c) = \tr &\Big(\rho^{\alpha}_{B_LC_R}\otimes\rho^{\beta}_{C_LA_R}\otimes \rho^{\gamma}_{A_LB_R} \\
&\times E_{A_LA_R}^{a} \otimes E_{B_LB_R}^{b} \otimes E_{C_LC_R}^{c}    \Big)
\end{split} \label{quantumtriangle}
\ee
where $\rho^{\xi}$ denote the quantum states distributed by the sources (density operators), and $\{E_{YY'}^{x}\}$ denote the measurements performed at each party and are given by positive operator valued measures (POVMs). Note that one should pay attention to the order of subsystems when computing \eqref{quantumtriangle}.

 For a given output cardinality $d$  we denote the set of all triangle-local distributions $\mathcal{L}_\triangle^{d}$ and the set of all triangle-quantum distributions $\mathcal{Q}_\triangle^{d}$. In general, characterizing these sets is a challenging problem. Outer approximations can be derived \cite{Chaves2015,wolfe2019inflation,wolfe2021inflation,Aberg2020}, but typically only provide loose bounds.

 For $d \geq 3$, it is known that $\mathcal{L}_\triangle^{d} \subset\mathcal{Q}_\triangle^{d}$, i.e. there exist quantum distributions that are nonlocal (see \cite{Boreiri2023} and references therein). For $d=2$, this is still an open question.

Known proofs of quantum nonlocality for the triangle network are of two types. The first are distributions that can be viewed as a clever embedding of a standard Bell test (e.g. CHSH) on the triangle network without inputs \cite{Fritz_2012}. The idea is to distribute effective inputs via two classical sources. Each party will broadcast the received inputs via their output. The third source distributes the entangled state as in the usual Bell test. See also \cite{Fraser,Weilenmann_2018,supic,Boreiri2023} for more examples. For this class of distributions, noise-robust methods have been recently developed \cite{Chaves2021}, leading to a first experiment \cite{Polino2023}. Due to their connection with standard Bell tests, these distributions can be implemented without the need for any entangled measurements. Hence, the nonlocality of these distributions originates solely from the use of an entangled state. 
 
 In contrast, there exist quantum distributions on the triangle network based on the judicious combination of entangled states and measurements, as first shown by Renou and colleagues \cite{Renou_2019}. It was recently shown that this distribution requires in fact the use of entangled measurements \cite{Sekatski2023}, hence demonstrating genuine network quantum nonlocality \cite{Supic2022}. So far, these nonlocality proofs are essentially restricted to an ideal noiseless scenario; see e.g. \cite{Renou_2019,renou2022network,Abiuso2022,Boreiri2023,Pozas2023}. That is, they consider a setting where the shared states $\rho^{\xi}$ are pure and the local measurements $\{E_{YY'}^{x}\}$ are projective. However, these proofs no longer apply when any realistic amount of noise is included\footnote{Note the semi-analytical methods of inflation \cite{wolfe2019inflation} can detect a specific instance of the RGB4 distribution, but the resulting noise-robustness is around $10^{-7}$ for white noise at the sources, so irrelevant from any practical perspective.}. On the other hand, numerical analysis based on machine learning suggests that the nonlocality of these quantum distributions is in fact robust to noise \cite{krivachy_neural_2020}, up to $\approx 10 \%$ white noise at the sources for some instances of the RGB4 distribution. Deriving noise-robust analytical proofs is therefore an important open problem, which we address in this work.

 \subsection{PTC distributions and rigidity}

A key class of distributions for our work is termed parity token counting (PTC). These are local distributions on the triangle network, obtained via the following model. Each source possesses a single token that it sends to either one of its connected parties with some given probability. For instance, the source $\alpha$ sends its token to Bob with probability $p_\alpha$ and to Charlie with probability $1-p_\alpha$. Then each party outputs the parity of the total number of received tokens. Therefore, consider that we define the binary random variable $t_{\alpha} \in \{0,1\}$ as the number of tokens $\alpha$ sends to Bob, $t_{\beta} \in \{0,1\}$ as the number of tokens  $\beta$ sends to Charlie, and $t_{\gamma} \in \{0,1\}$ as the number of tokens  $\gamma$ sends to Alice. Then for example, Alice receives $t_{\gamma}$ from the source $\gamma$ and $t_{\beta} \oplus 1$ from the source $\beta$ and her response function would be $a(\beta, \gamma) = t_{\gamma} \oplus t_{\beta} \oplus 1$, where $\oplus$ denotes the sum modulo two.

The distribution $P(a,b,c)$ resulting from the above PTC model has binary outputs ($d=2$). By construction, any such distribution satisfies 
 \be\label{eq: PTC condition}
 a \oplus b \oplus c =1,
 \ee 
which we refer to as the PTC condition. 

Let us now explore the reverse link. Starting from a distribution that satisfies the PTC condition, what can we infer about the underlying classical model? Interestingly, PTC distributions demonstrate a form of rigidity \cite{Boreiri2023}; note that this property of rigidity was initially discussed for token counting models \cite{renou2022network}. Rigidity means that a PTC distribution can only be obtained via a PTC model (up to irrelevant relabellings). That is, for any local model leading to a PTC distribution, there exists for each source a token function
\be \label{eq: PTC rigidity}\begin{split}
&T_\xi: \xi \to t_\xi \in \{0,1\} \qquad \text{such that}\\
&x(\xi,\xi')= t_\xi \oplus t_{\xi'}\oplus 1,
\end{split}
\ee
for the sources $\xi , \xi' \in \{\alpha, \beta, \gamma\}$ and $x \in \{a,b,c\}$ representing the party connected to the sources $\xi , \xi'$.
In addition, if the probability distribution is such that $E_a, E_b, E_c\neq 0$ for $E_x = \text{Pr}(x=0)-\text{Pr}(x=1)$, the token functions can be chosen to fulfill
\be\label{eq: token values}
p_{\alpha}(t):=\text{Pr}(t_\alpha=t) = \frac{1}{2}\left(1+(-1)^t\sqrt{\left|\frac{E_b E_c}{E_a}\right|}\right)
\ee
and similar equalities of the other sources. This means that there is essentially only a unique local model over the triangle which can simulate such a distribution. In~\cite{Boreiri2023} one can find  a simple characterization of the local PTC set $\mathcal{L}_{\triangle|PTC}^2$, containing all triangle-local distributions that satisfy the PTC condition of Eq.~\eqref{eq: PTC condition}. A graphical representation of $\mathcal{L}_{\triangle|PTC}^2$ is given in Fig.~\ref{fig:LPTC}, we find that it covers only about 12.5\% of the volume\footnote{We have computed the volume of the set by uniformly sampling $10^6$ points $P$ from the PTC slice $\mathcal{P}_{PTC}^2$ and verifying for each point if it admits a local model or not. For symmetric distributions  $E_A=E_B=E_C$ we find that the set $\mathcal{L}_{\triangle|PTC}^2$ is characterized by a simple inequality $P(1,1,1)\geq\frac{1}{4}$.} of $\mathcal{P}_{PTC}^2$ -- the set of binary probability distributions $P(a,b,c)$ fulfilling Eq.~\eqref{eq: PTC condition}.
\begin{figure}[h]
    \centering
\includegraphics[width=0.95\columnwidth]{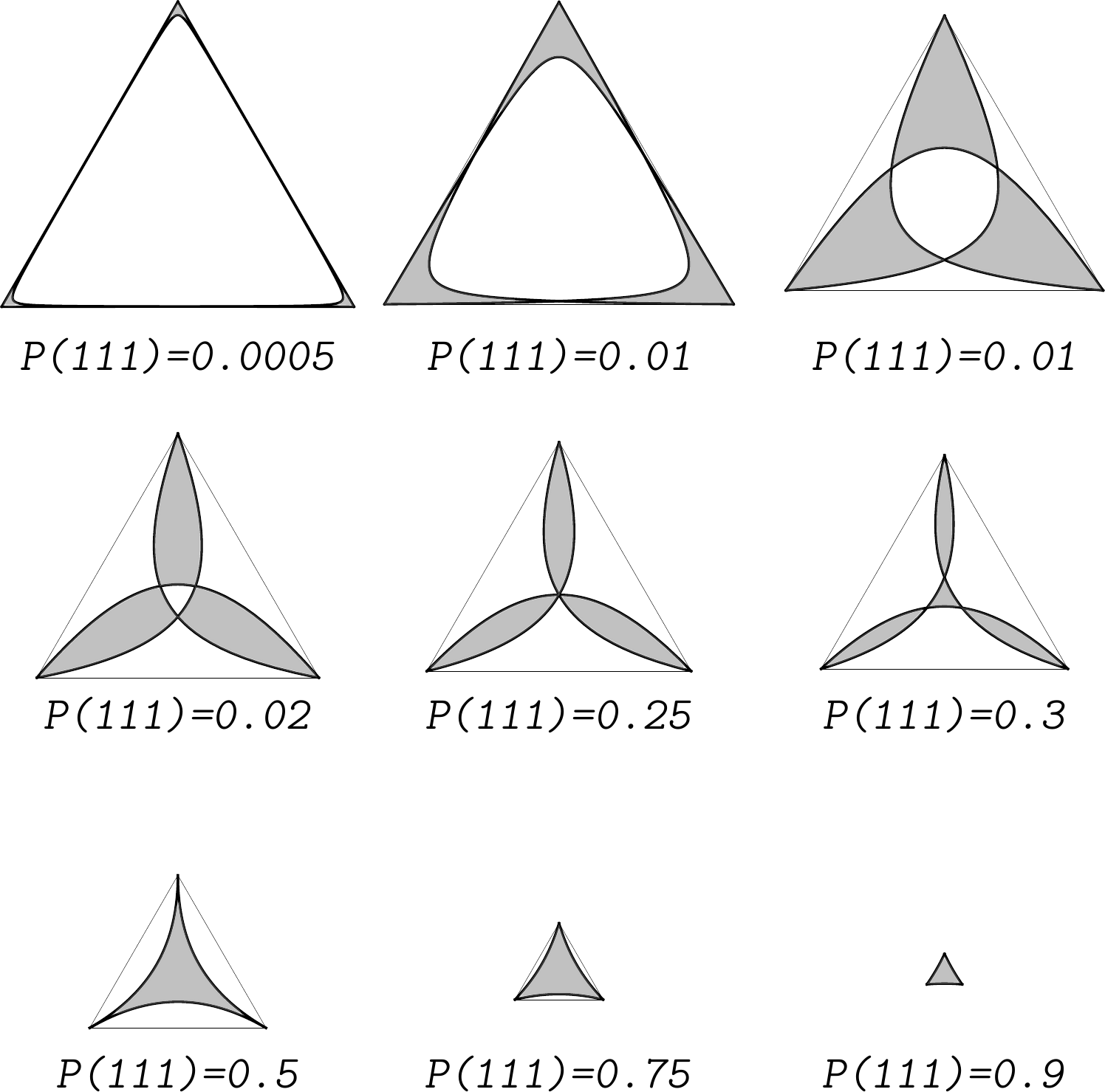}
    \caption{A layer-by-layer representation of the set $\mathcal{L}_{\triangle|PTC}^2=\mathcal{Q}_{\triangle|PTC}^2 $, containing all triangle-local (or equivalently quantum) distributions $P(a,b,c)$ with binary outcomes,
    satisfying the deterministic PTC constraint $\text{Pr}(a\oplus b\oplus c =1 )=1$. Each layer depicts the set for a given value of $P(a=b=c=1)$. The distributions inside the set can be represented by the vector $\bm v= \big(P(1,0,0),P(0,1,0),P(0,0,1)\big)$, with the remaining nonzero probability $P(1,1,1)$ given by normalization. Each  plot in the figure is a slice of the set for a  fixed value of $P(1,1,1)$, i.e. it is a cut of the set with the plane define by the equation $\bm v \cdot (1,1,1) =1-P(1,1,1)$. In each layer, the outer triangle delimits the set of valid probabilities. A 3D representation of the set $ \mathcal{L}_{\triangle|PTC}^2$ in printable format can be obtained from the authors upon request.}
    \label{fig:LPTC}
\end{figure}

Remarkably, the PTC rigidity, albeit formulated differently, is also true for any quantum PTC distribution \cite{Sekatski2023}. In particular, this implies that in the slice of the probability set where the parity condition in Eq.~\eqref{eq: PTC condition} holds (ideal case without noise), there is no separation between the quantum and local sets, i.e. $\mathcal{Q}_{\triangle|PTC}^2 = \mathcal{L}_{\triangle|PTC}^2$.

\section{Approximate PTC rigidity for local distributions}
\label{sec: approximate PTC}

Let us now consider a local distribution $P(a,b,c)$ for which the parity condition only holds approximately, i.e.
\be \label{eq: approx rigid}
\text{Pr}(a \oplus b \oplus c = 1) \geq 1- \ve.
\ee
Intuitively, one expects the underlying local model to verify some approximate form of rigidity. This is what we now establish.

Since the response functions are deterministic\footnote{Note that any measurement that is not described by a deterministic response function can be decomposed as an additional local source of randomness followed by a deterministic response function. In this case, the randomness source can be merged with one of the sources present in the network, and all of the following arguments can be applied at the level of the merged sources.}, there is a set 
\be
\Lambda =\{\bm \xi | a(\beta,\gamma)\oplus b(\alpha,\gamma)\oplus c(\alpha,\beta)=1\} 
\ee
containing all the values $\bm \xi =(\alpha,\beta,\gamma)$ for which the parity condition in Eq.~\eqref{eq: PTC condition} is true. By assumption we have $\text{Pr}(\bm \xi \in \Lambda) \geq 1-\ve$.\\

Let us now take any value $\gamma_*$ for which $\text{Pr}(\bm \xi \in \Lambda|\gamma=\gamma_*)\geq 1-\ve$. Such a value must exist, since this probability cannot be strictly lower than its average on all $\gamma$. For the chosen value of $\gamma$, we define the token function to be 
\be
t_{\gamma_*}=1.
\ee
Next, take  $G_* = \{(\alpha,\beta)| (\alpha,\beta,\gamma_*)\in \Lambda\}$
to be the set of all pairs $(\alpha,\beta)$ for which the parity condition holds when $\gamma=\gamma_*$. On this set we now define the token functions $t_\alpha$ and $t_\beta$ as 
\be\begin{split}
t_\alpha &= b(\alpha,\gamma_*)\\
t_\beta &= a(\beta,\gamma_*).
\end{split}
\ee
Let us check that for $\gamma=\gamma_*$ and any values $\alpha$ and $\beta$ in $G_*$, the PTC rigidity of Eq.\eqref{eq: PTC rigidity} holds. By construction we have $t_\beta\oplus t_{\gamma_*}\oplus 1 = t_\beta = a(\beta,\gamma_*)$
and  $t_\alpha\oplus t_{\gamma_*}\oplus 1 = b(\alpha,\gamma_*)$, while
\be
t_\beta \oplus t_\alpha \oplus 1 = a(\beta,\gamma_*) \oplus b(\alpha,\gamma_*)\oplus 1 = c(\alpha,\beta)
\ee
follows from the parity condition at $(\alpha,\beta,\gamma_*)\in\Lambda$. In addition, we have a lower bound on the probability that $(\alpha,\beta)$s are sampled in $G_*$
\be
\text{Pr}\big((\alpha,\beta)\in G_*\big) =\text{Pr}(\bm \xi \in \Lambda|\gamma=\gamma_*)\geq 1-\ve.
\ee

Finally, consider the set 
\be
\Omega =\{\bm \xi | (\alpha,\beta)\in G_* \text{ and } \bm \xi \in \Lambda\}.
\ee
which has the probability 
\be\begin{split}
\text{Pr}(\bm \xi \in \Omega)
& \geq 1- \text{Pr}(\bm \xi \notin \Lambda) - \text{Pr}((\alpha,\beta)\notin G_*)\\
&\geq 1-2\ve.
\end{split}
\ee
For any point $(\alpha,\beta,\gamma)$ from this set, we can define the prolongation of the token function $t_\gamma$ as  
\begin{equation}\label{eq: gamma prolong}\begin{split}
    t_\gamma^{(\alpha,\beta)} &= 1 \oplus a(\beta,\gamma)\oplus a(\beta,\gamma_*)\\
 &= 1\oplus b(\alpha,\gamma)\oplus b(\alpha,\gamma_*),    
\end{split}
\end{equation}
where the equality of the two definitions follows from the fact that both points $(\alpha,\beta,\gamma)$ and $(\alpha,\beta,\gamma_*)$ are inside $\Lambda$ by construction. It is also easy to see that this value $t_\gamma$ satisfies the PTC rigidity, i.e. $t^{(\alpha,\beta)}_{\gamma}\oplus t_\beta \oplus 1 = a(\beta,\gamma)$ and the other equations, for the variables value $(\alpha,\beta,\gamma)$ (see Appendix \ref{app: approximate PTC}). The problem however is that the prolongations of $t_\gamma$ must not necessarily agree for different choices of $(\alpha,\beta)$. Hence this strategy can not be used to define the token function $t_\gamma$ on the whole set $\Omega$. Nevertheless, we show in  Appendix~\ref{app: approximate PTC}, that there exists a subset $\Upsilon \subset \Omega$ such that
\be \label{eq: fe}
\begin{split}
 \text{Pr}(\bm \xi \in \Upsilon)&\geq 1- \fe \geq 1-3\ve \quad \text{with}\\
 1-\fe &= \begin{cases}
 \frac{1}{2}-\ve +\frac{1-4\ve}{2 \sqrt{1-2 \ve }} & \ve \leq \frac{1}{4}\\
\frac{1-2\ve}{2}  & \ve >\frac{1}{4}
\end{cases} 
\end{split}
\ee
on which the token function $t_\gamma$ can be prolonged consistently. The key idea is to use the fact that the definitions in Eq.~\eqref{eq: gamma prolong} must agree for any two points in $\Omega$ with the same $\alpha$ or $\beta$. This demonstrates the following result holds.\\

\textbf{Result 1. (Approximate PTC rigidity)} \textit{For any triangle-local model leading to a distribution $P(a,b,c)\in \mathcal{L}_{\triangle}^2$ that fulfills the approximate parity condition $\text{Pr}(a \oplus b \oplus c = 1) \geq 1- \ve$, there exists a subset $\Upsilon$ of local variable values $\bm \xi=(\alpha,\beta,\gamma)$ with} $Pr(\bm \xi \in \Upsilon)\geq 1-\fe$ \textit{(see Eq.~\ref{eq: fe}), on which one can define the token functions $T_\alpha,T_\beta,T_\gamma$ that satisfy the PTC rigidity $a(\beta,\gamma)=T_\beta\oplus T_{\gamma}\oplus 1$ etc.}\\

For any triangle-local model, the set $\Upsilon$ gives rise to the subnormalized distribution  $P(a,b,c;\bm\xi\in \Upsilon)$ representing the probability of the event "$\xi=(\alpha,\beta,\gamma)\in \Upsilon$ and $a(\beta,\gamma) =a, b(\alpha,\gamma)=b, c(\alpha,\beta)=c$" for the output values $a,b,c$. These probabilities  satisfy $\sum_{a,b,c} P(a,b,c;\bm\xi\in \Upsilon) \geq 1-\fe$. For simplicity let us assume that this bound is tight (otherwise one can repeat the following discussion with a second parameter $\text{f}' \leq \fe$ and reach the same conclusions). It follows that 
\be\label{eq: P'}
P(a,b,c;\bm\xi\in \Upsilon) +\fe P'(a,b,c) = P(a,b,c),
\ee
where $P'(a,b,c)$ is some\footnote{Note that we know the values $P'(a,b,c)$ on the outcomes with wrong parity and will exploit this fact later.} (non necessarily local) distribution.\\

Now, let us extend the token functions $T_\xi$ to all the values of the local variables that do not appear in $\Upsilon$, this can be done by e.g. assigning $t_\xi = 0$ to all such values. With the token functions defined everywhere, we can now define the corresponding extended probability distribution as follows
\begin{align}
&\widetilde P(a,b,c)
\\&= \mathds{E}\Big(P_{\scriptscriptstyle ptc}(a|t_\beta ,t_\gamma)P_{\scriptscriptstyle ptc}(b|t_\alpha ,t_\gamma)P_{\scriptscriptstyle ptc}(c|t_\alpha ,t_\beta) \Big) \nonumber
\end{align}
where the outputs are obtained with the PTC response functions
\be\label{eq: PTC response}
P_{\scriptscriptstyle ptc}(x|t_\xi ,t_{\xi '}) =\begin{cases}
1 & x= t_\xi \oplus t_{\xi '}\oplus 1\\
0 & x\neq t_\xi \oplus t_{\xi '}\oplus 1
\end{cases}.
\ee
By construction, this distribution is triangle-local and PTC. Furthermore, we know that  $P(a,b,c)$ and $\widetilde P(a,b,c)$ agree on the subset $\bm \xi \in \Upsilon$, hence the identity 
\be\label{eq: tilde P'}
P(a,b,c;\bm\xi\in \Upsilon) +  \fe \widetilde P'(a,b,c) =  \widetilde P(a,b,c)
\ee
holds for some PTC (non-necessarily local) distributions $\widetilde P'$. Combining Eqs.~(\ref{eq: P'},\ref{eq: tilde P'}) we conclude that the observed local distribution must satisfy 
$P(a,b,c) = \widetilde P(a,b,c) - \fe \widetilde P'(a,b,c) +\fe P'(a,b,c)$,
which can be nicely expressed in terms of total variation distance
\be
\delta(P,\widetilde P)=\frac{1}{2}\sum_{a,b,c}|P(a,b,c) -\widetilde P(a,b,c)|\leq \fe.
\ee

We have thus shown that any local distribution $P(a,b,c)$ satisfying the approximate parity condition of Eq.~\eqref{eq: approx rigid}, must be close to the local PTC set
\be
\delta(P,\mathcal{L}^2_{\triangle|PTC}) = \min_{\widetilde P\in \mathcal{L}_{\triangle|PTC}^2} \delta(P,\widetilde P)\leq \fe
\ee
In turn, the approximate parity condition is equivalent to a bound on the total variation distance between $P$ and the PTC slice of the probability set
$\delta(P, \mathcal{P}_{PTC}^2) = \min_{Q\in \mathcal{P}_{PTC}^2} \delta(P,Q) \leq \ve.$ Therefore, we have shown the following result. \\

\textbf{Result 2. (Relaxation of the local set $\mathcal{L}_\triangle^2$)} \textit{Any triangle-local distribution $P(a,b,c) \in \mathcal{L}_\triangle^2$ that is close to the PTC slice of the probability set $\delta(P, \mathcal{P}_{PTC}^2)\leq \ve$ is also close to the local (also triangle-local) PTC set} $\delta(P,\mathcal{L}^2_{\triangle|PTC})\leq \fe$ \textit{(see Eq.~\ref{eq: fe}).}\\

\begin{figure}
\includegraphics[width=\columnwidth]{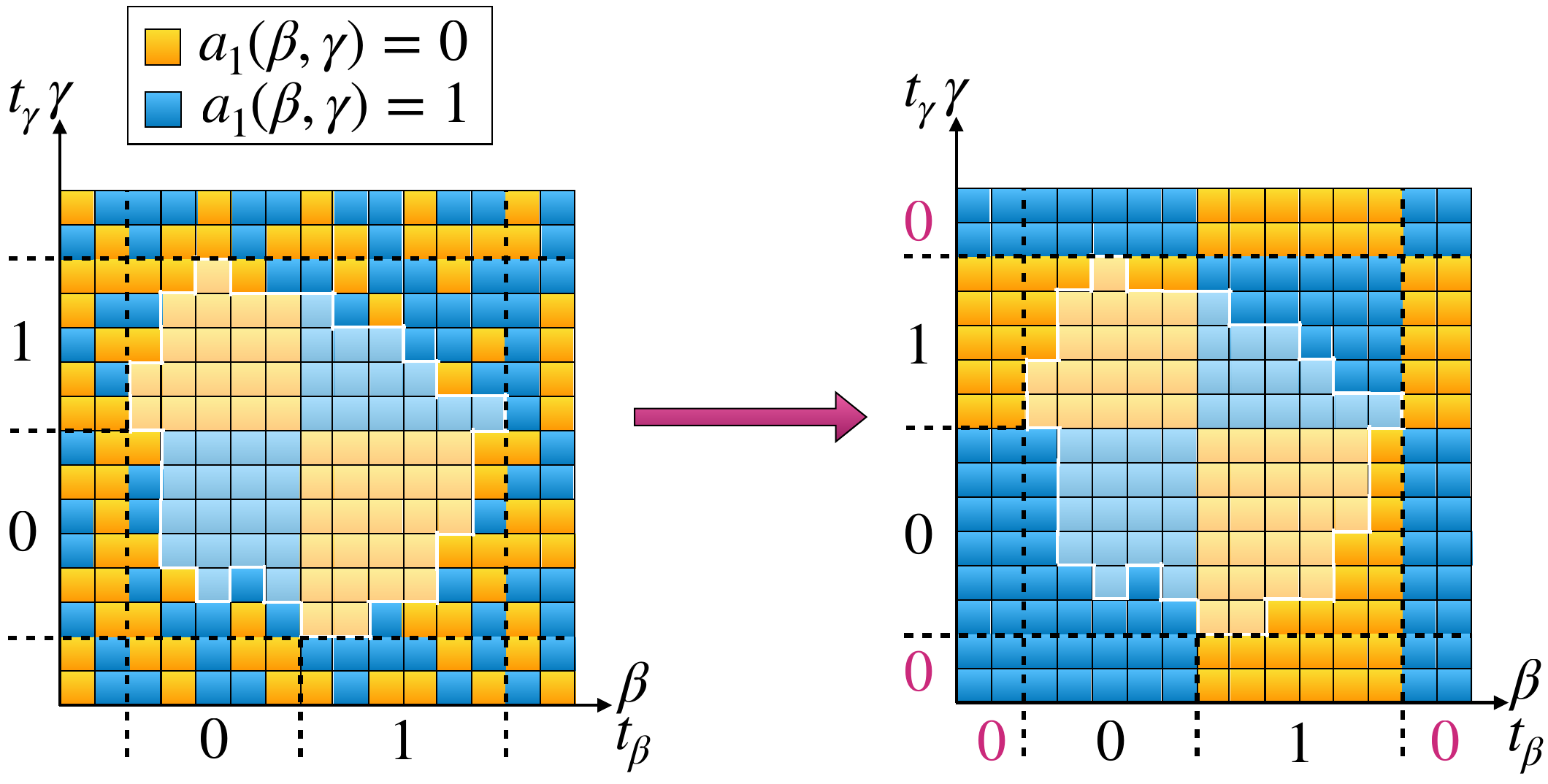}
    \caption{A graphical representation of the construction used in the proof of approximate PTC rigidity. The coordinates along the axes represent different values of the local variables $\beta$ and $\gamma$, here discretized to 15 possible values each. The color of each square gives the value of Alice's outcome $a(\beta,\gamma)=0$ (yellow) and $a(\beta,\gamma)=1$ (blue). On the left, the lighter region depicts the set $\Upsilon$ on which the binary token functions $t_\beta$ and $t_\gamma$ are defined (given below the axes) and satisfy $a(\beta,\gamma)=t_\beta\oplus t_\gamma \oplus 1$. On the right, the token functions are extended on the whole set of local variables by setting them to 0 outside of $\Upsilon$. The values obtained by applying the PTC response function $ a(\beta,\gamma)=t_\beta\oplus t_\gamma \oplus 1$ (right) are not guaranteed to match the observed values $a(\beta,\gamma)$ (left) outside of $\Upsilon$.  
    }
    \label{fig:RelabelPTC}
\end{figure}

In words, the approximate rigidity allows us to define a cone of distributions $\widehat{\mathcal{L}}^2_{\triangle}$ that are accessible when going away from the PTC slice, as illustrated in Fig.~\ref{fig:result2}. Loosely speaking -- it contains triangle-local distributions at a distance $\ve$ to the PTC slice that must also be close enough to the local PTC set such that we can not move more than a distance $\ve$ \textit{inside} the slice. Any distribution outside of the cone is guaranteed to be nonlocal, i.e. we have established an outer approximation of the local set $\mathcal{L}_\triangle^2 \subset \widehat{\mathcal{L}}^2_{\triangle}$. It is an open question whether all quantum distributions are also contained in this cone.

To illustrate Result 2 we discuss the nonlocality of the noisy W distributions (see Appendix \ref{app: noisy W}).

\begin{figure}
    \centering
    \includegraphics[width=0.9\columnwidth]{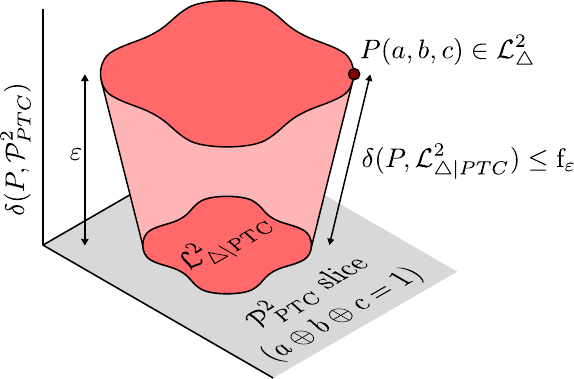}
    \caption{A graphical representation of Result 2. Consider any triangle-local distribution $P(a,b,c)\in \mathcal{L}^2_\triangle$ which is at a distance $\varepsilon= \delta(P,\mathcal{P}_{PTC}^2)$ to the PTC slice containing all probability distributions satisfying $\text{Pr}(a\oplus b\oplus c \neq 1)=0$. Result 2 implies that $P(a,b,c)$ must also be "close enough" $\delta(P,\mathcal{L}_{\triangle|PTC}^2)\leq \fe$ to the local PTC set $\mathcal{L}^2_{\triangle|PTC}$ of which we have full characterization. Here, $\delta$ is the total variation distance, and the function $\fe\leq 3\ve$ is defined in Eq.~\eqref{eq: fe}.}
    \label{fig:result2}
\end{figure}

\section{Noise-robust nonlocality in $\mathcal{Q}^4_\triangle$}

\label{sec: TBSM}
We now consider a particular family of four-outcome quantum distributions on the triangle. Those are obtained from the following quantum models. Each source distributes a two-qubit state
\be\label{eq: pure state}
\ket{\psi} = \lambda_0\ket{01} + \lambda_1\ket{10},
\ee

with real and positive $\lambda_0^2+\lambda_1^2=1$, to the neighboring parties. Each party performs the same projective measurement $\{\Pi_{x_1,x_2}=\ketbra{\phi_{x_1,x_2}}\}$ on the two qubits they receive from the neighboring sources. The POVM elements are given by projectors on the states:
\be
\begin{split}
    \ket{\phi_{1,0}} &= u \ket{0,1} + v\ket{1,0}\\
    \ket{\phi_{1,1}} &= v \ket{0,1} - u\ket{1,0}\\
    \ket{\phi_{0,0}} &= w \ket{0,0} + z\ket{1,1}\\
    \ket{\phi_{0,1}} &= z \ket{0,0} - w\ket{1,1}
\end{split}
\ee
with real positive $u,v,w,z$ satisfying $u^2+v^2=w^2+z^2=1$. Here, the element $\ket{\phi_{x_1, x_2}}$ corresponds to the 2-bit output ${x_1,x_2}$ for the party $x \in \{a,b,c\}$. The resulting family of quantum distributions $P(a_1, a_2, b_1, b_2, c_1, c_2)$ is termed "tilted Bell state measurements" (\name) distribution, since the local measurements can be viewed as tilted Bell state measurements. Elements of such a distribution are given by
\be\label{eq: BUBS distribution}
\begin{split}
P(1,i,1,j,1,k) &= (\lambda_0^3 u_i u_j u_k + \lambda_1^3 v_i v_j v_k)^2\\
P(1,i,0,j,0,k) &= (\lambda_0 \lambda_1^2 u_i w_j z_k + \lambda_1 \lambda_0^2 v_i w_k z_j)^2 \,\circlearrowleft
\end{split}
\ee
where $u_0=u,\, v_0=v\, ,w_0=w\, ,z_0=z,\,u_1 = v, \, v_1 = -u, \, w_1 = z$ and $z_1 = -w$. Here $\circlearrowleft$ means that the equation is valid up to cyclic permutations of the parties.

By setting $w=1-z=0$ one recovers the RGB4 family of distributions \cite{Renou_2019}, while merging the outcomes $(0,0)$ and $(0,1)$ leads to the three-output family of quantum distributions introduced in \cite{Boreiri2023}. 

As a warm-up, we now present a proof of nonlocality of the distributions in Eq.~\eqref{eq: BUBS distribution} for a subset of parameter values $(\lambda_0, u, w)$. It has the advantage of streamlining the proof techniques used in \cite{Renou_2019} and subsequently \cite{Boreiri2023}, while relying on the same basic idea. This will be helpful when we move to the noisy case.

\subsection{ Nonlocality of the noiseless TBSM distributions}
\label{sec: noiseless}

The nonlocality proof proceeds by contradiction, so let us assume that the distribution is local. It is easy to see that when the second output is discarded, the resulting binary distributions
\be
P(a_1,b_1,c_1) =\sum_{a_2,b_2,c_2}  P(a_1,a_2,b_1,b_2,c_1,c_2)
\ee
satisfy the parity condition $a_1\oplus b_1\oplus c_1=1$. By rigidity, it follows that for any underlying local model there exist the binary token values $t_\alpha, t_\beta, t_\gamma$ that explain the distribution 
\be
P(\bm x_1) = \mathds{E} \Big(P_{\scriptscriptstyle PTC}(\bm x_1|\bm t) \Big),
\ee
where $\mathds{E}$ is the expected value over $\bm t$ and we introduced the compact notation $\bm x_i=(a_i,b_i,c_i)$, $\bm t =(t_\alpha, t_\beta, t_\gamma)$ and $P_{\scriptscriptstyle PTC}(\bm x_1|\bm t) =
P_{\scriptscriptstyle ptc}(a_1|t_\beta ,t_\gamma)P_{\scriptscriptstyle ptc}(b_1|t_\alpha ,t_\gamma)P_{\scriptscriptstyle ptc}(c_1|t_\alpha ,t_\beta),$ with the response functions of Eq.~\eqref{eq: PTC response}.
Furthermore, if $E_{a_1},E_{b_1},E_{c_1}\neq 0$ the token distribution  $p(\bm t )= p_\alpha(t_\alpha)p_\beta(t_\beta)p_\gamma(t_\gamma)$ is known exactly via Eq.~\eqref{eq: token values}. \\

This conclusion is, of course, independent of whether we discard the outcomes $\bm x_2$ or not. Hence, for any local model resulting in the \name \ distribution, we can define the random variable $\bm t$, indicating the behavior of the tokens. Using the chain rule, we have 
$P(\bm x_1,\bm x_2)= \sum_{\bm t} p(\bm t) P(\bm x_2|\bm t) P(\bm x_1|\bm t , \bm x_2)$. From the rigidity, we can infer that $\bm x_1$ is solely a function of $\bm t$, leading to
\be
\begin{split}
    P(\bm x_1,\bm x_2)= \sum_{\bm t} p(\bm t) P_{\scriptscriptstyle PTC}(\bm x_1|\bm t) P(\bm x_2|\bm t).
\end{split}
\ee
Here, all the distributions $P(\bm x_2|\bm t)\in \mathcal{L}_\triangle^2$ must be triangle-local. Furthermore, these distributions are not independent, since the output of a party (e.g. $a_2$) can not be influenced by the source that is not connected to it (that is by $\alpha$ and the token value $t_\alpha$). This observation implies the following constraints on the marginal distributions $P(a_2|\bm t) = \sum_{b_2,c_2} P(\bm x_2|\bm t)$
\be\label{eq: network cons}
P(a_2|t_\alpha = 0,t_\beta,t_\gamma)=P(a_2|t_\alpha=1,t_\beta,t_\gamma) \quad \circlearrowleft
\ee
and similar constraints for all permutations of parties, as indicated by the symbol $\circlearrowleft$. The above discussion can be summarized as follows.\\

\textbf{Result 3. (Nonlocality of TBSM distributions)}\textit{ For any triangle-local distribution $P(\bm x_1,\bm x_2)\in \mathcal{L}_\triangle^4$ satisfying $a_1\oplus b_1\oplus c_1 =1$ , there exist eight probability distributions $P(\bm x_2|\bm t)$ such that}
\be\label{eq: nonloc pure}\begin{split}
& (3.i) \,P(\bm x_1,\bm x_2) =\sum_{\bm t} p(\bm t) P_{\scriptscriptstyle PTC}(\bm x_1|\bm t) P(\bm x_2|\bm t)\\
&(3.ii) \,P(a_2|t_\alpha = 0,t_\beta,t_\gamma)=P(a_2|t_\alpha=1,t_\beta,t_\gamma) \quad \circlearrowleft\\
&(3.iii)\,P(\bm x_2|\bm t) \in \mathcal{L}_\triangle^2 \qquad \forall \bm t
\end{split}
\ee
\textit{where $P_{\scriptscriptstyle PTC}$ is the parity token counting response function and the token distribution $p(\bm t )= p_\alpha(t_\alpha)p_\beta(t_\beta)p_\gamma(t_\gamma)$ is uniquely determined from $P(\bm x_1)$ via Eq.~\eqref{eq: token values} (except degenerate cases where it is not unique).}
\\

If we were able to show that these conditions cannot be satisfied together, we could conclude that the distribution $P(\bm x_1,\bm x_2)$ is nonlocal. Yet, the locality condition $(3.iii)$ here is very challenging to formalize since we do not have a good characterization of the local set $\mathcal{L}_\triangle^2$. A simple solution is to relax it and keep the other two conditions. This relaxation gives rise to a simple linear program (LP) that must be feasible for any triangle-local $P(\bm x_1,\bm x_2)$ with $a_1\oplus b_1\oplus c_1 =1$, that we give in Appendix \ref{app: LP excact PTC}.

\begin{figure}
    \centering
    \includegraphics[width=0.9\columnwidth]{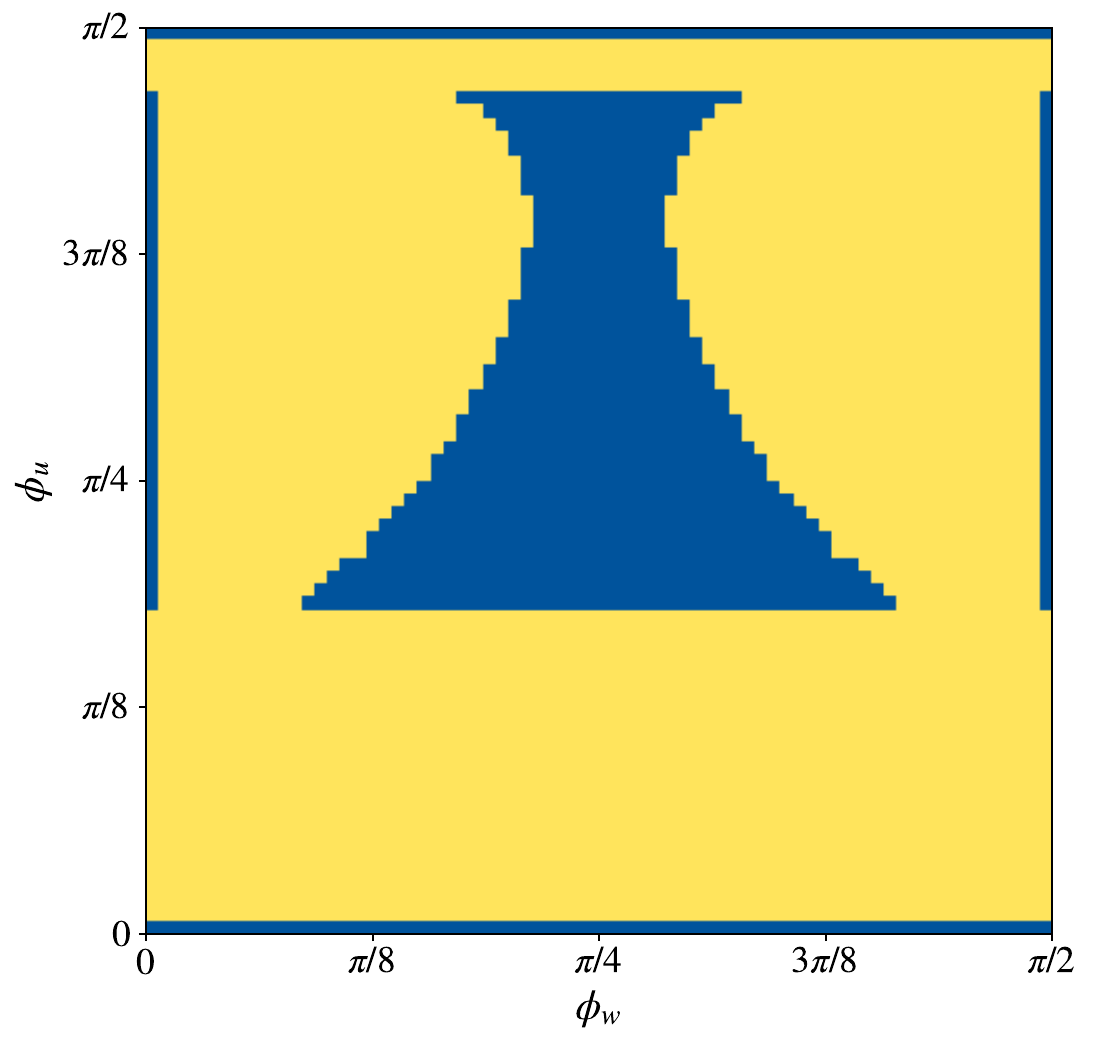}
    \caption{\textbf{Nonlocality of noiseless TBSM distributions.} Nonlocal region depicted in yellow. Here distributions are parametrized via angles: $u=\cos(\phi_u)$, $w=\cos(\phi_w)$. Parameter: $\lambda_0^2= 0.22$.}
    \label{fig:nonlocal_set}
\end{figure}

In Fig.~\ref{fig:nonlocal_set}, we present the results of the LP for $\lambda_0^2=0.22$. One notes that Result 3 can be used to prove the nonlocality of the noisy distribution if the noise does not perturb the parity of the $\bm x_1$ outcomes. For example, the local dephasing noise given by the channel $\cE_d: \ketbra{\psi} \mapsto \varrho$ with
\be\label{eq: dephasing1}
\varrho= (1-d) \ketbra{\psi} + d \left(\lambda_0^2 \ketbra{01} +  \lambda_1^2 \ketbra{10}
\right)\ee
is of this type. In Appendix \ref{app: LP excact PTC} we give the LP which finds the minimal amount of dephasing noise such that the distribution becomes triangle-local. Its results are presented in Fig.~\ref{fig:Dephasing} for $\lambda_0^2=0.22$.

\begin{figure}
    \centering
    \includegraphics[width=\columnwidth]{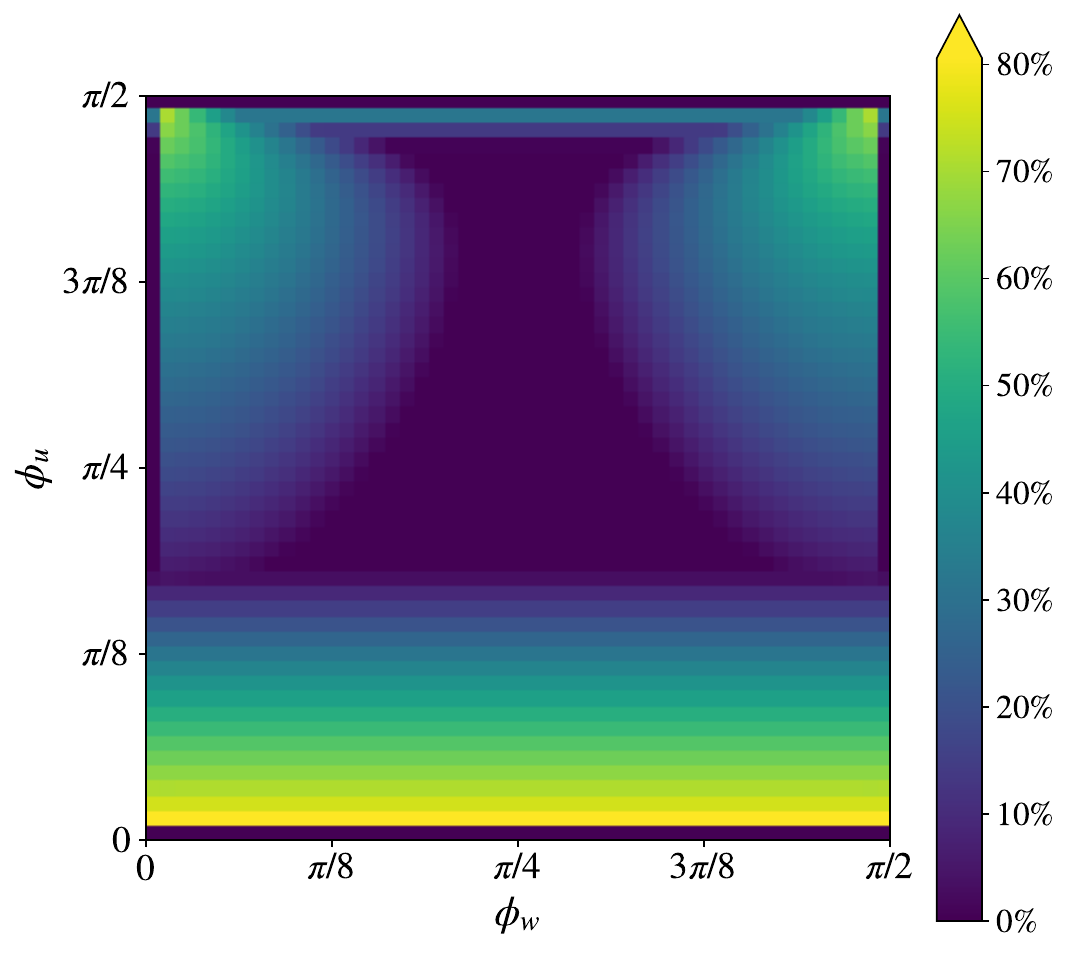}
    \caption{\textbf{Dephasing Noise Robustness.} The amount of local dephasing noise $d$ in Eq.~\eqref{eq: dephasing1} below which the TBSM distribution is guaranteed to be nonlocal. Parameters: $u=\cos(\phi_u)$, $w=\cos(\phi_w)$ and $\lambda_0^2= 0.22$.}
    \label{fig:Dephasing}
\end{figure}

Nevertheless, this still only applies to a slice of the quantum set $\mathcal{Q}_\triangle^4$. In the next Section, we derive an approach that singles out the nonlocality of a full-measure subset of triangle quantum distributions. Interestingly, we observe that the largest robustness to white noise, as well as for total-variation distance, is obtained for the same parameters $\lambda_0^2 = 0.22, \phi_u=0.42, \phi_w=0$, with $u=\cos(\phi_u)$ and  $w=\cos(\phi_w)$.

\subsection{Nonlocality of the noisy TBSM distributions}

\label{sec: WN}
Our main idea now is pretty simple -- adjust the LP-based nonlocality proof of the previous Section to the setting of approximate PTC rigidity discussed in Section \ref{sec: approximate PTC}.

Indeed, let us consider any triangle-local distribution 
\begin{align}\nonumber
&P(\bm x_1, \bm x_2) \\
&=\mathds{E}\big(P_A(a_1,a_2|\beta,\gamma)P_B(b_1,b_2|\alpha,\gamma) P_C(\dots) \big) \nonumber
\end{align}
such that 
$\text{Pr}(a_1\oplus b_1\oplus c_1=1)\geq 1-\ve$. By approximate PTC rigidity, we know that there exists a local PTC distribution $\widetilde P(\bm x_1)$ such that 
$\delta(P(\bm x_1),\widetilde P(\bm x_1))\leq \fe.$
This distribution was constructed by applying the token functions $T_\xi$ to the sources, and replacing the detectors with the parity token counting response functions (see Fig. \ref{fig:RelabelPTC}). This construction does not require us to discard the outcomes $\bm x_2$, and we can define a full distribution 
\begin{align}
&\widetilde P(\bm x_1,\bm x_2)
\\ 
&= \mathds{E}\big(P_{\scriptscriptstyle PTC}(\bm x_1|\bm t) P_A(a_2|\beta,\gamma) P_B(b_2|\alpha,\gamma) \dots). \nonumber
\end{align}
The distribution $\widetilde P$ is triangle-local, by construction it agrees with the initial distribution on outcomes $\bm x_2$ ($\widetilde P(\bm x_2)= P(\bm x_2)$), and is PTC local for the outcomes $\bm x_1$. Furthermore, we know that the two distributions match $P(\bm x_1,\bm x_2; \bm \xi \in \Upsilon) = \widetilde P(\bm x_1,\bm x_2; \bm \xi \in \Upsilon)$ on the subset $\Upsilon$ of local variable values with $\text{Pr}(\bm \xi \in \Upsilon)\geq 1-\fe$, which is equivalent to a bound on their total variation distance
\be
\delta(P(\bm x_1,\bm x_2),\widetilde P(\bm x_1,\bm x_2)) \leq \fe.
\ee
We have thus reached the following conclusion.\\

\textbf{Result 4. (Nonlocality of noisy TBSM distribution)}\textit{ For any triangle-local distribution $P(\bm x_1,\bm x_2)\in \mathcal{L}_\triangle^4$ satisfying $\text{Pr}(a_1\oplus b_1\oplus c_1 =1)\geq 1-\ve$, there exists a triangle-local distribution $\widetilde P(\bm x_1,\bm x_2)$ such that }
\be\begin{split} 
& (4.i) \, \widetilde P(\bm x_1,\bm x_2) = 0 \quad \text{for} \quad a_1\oplus b_1\oplus c_1\neq 1\\
&(4.ii) \,  \delta(P(\bm x_1,\bm x_2),\widetilde P(\bm x_1,\bm x_2)) \leq \fe.\\
&(4.iii)\, \sum_{\bm x_1} \widetilde P(\bm x_1,\bm x_2) = \sum_{\bm x_1} P(\bm x_1,\bm x_2) 
\end{split}
\ee\\

In turn, the distribution $\widetilde P$ satisfies the premises of Result 3. We could thus verify the nonlocality of $P$ by showing that no distribution $\widetilde P$ compatible with both Result 3 and 4 exists. A complication here is that the relation~\eqref{eq: token values} between the distribution $\widetilde P(\bm x_1)$ and the underlying token distribution $p(\bm t)$ is nonlinear. It is thus easier to introduce the token distributions $p_\alpha, p_\beta, p_\gamma$ as additional variables rather than computing them from $\widetilde P$. That way we can combine results 3 and 4 into the following corollary.\\

\textbf{Corollary of Result 4.}\textit{ For any triangle-local distribution $P(\bm x_1,\bm x_2)\in \mathcal{L}_\triangle^4$ satisfying $\text{Pr}(a_1\oplus b_1\oplus c_1 =1)\geq 1-\ve$, there exist probability distributions $P(\bm x_2|\bm t)$, $p_\xi(t_\xi)$ and $\widetilde P(\bm x_1,\bm x_2)$ such that $(4.ii),(4.iii),(3.ii),(3.iii)$ and 
\be \label{eq: Corrolary4}
(*) \,\widetilde P(\bm x_1,\bm x_2) =\sum_{\bm t} p(\bm t) P_{\scriptscriptstyle PTC}(\bm x_1|\bm t) P(\bm x_2|\bm t)
\ee
for $p(\bm t) = p_\alpha(t_\alpha)p_\beta(t_\beta)p_\gamma(t_\gamma)$.} \\

\begin{figure}
    \centering
    \includegraphics[width=\columnwidth]{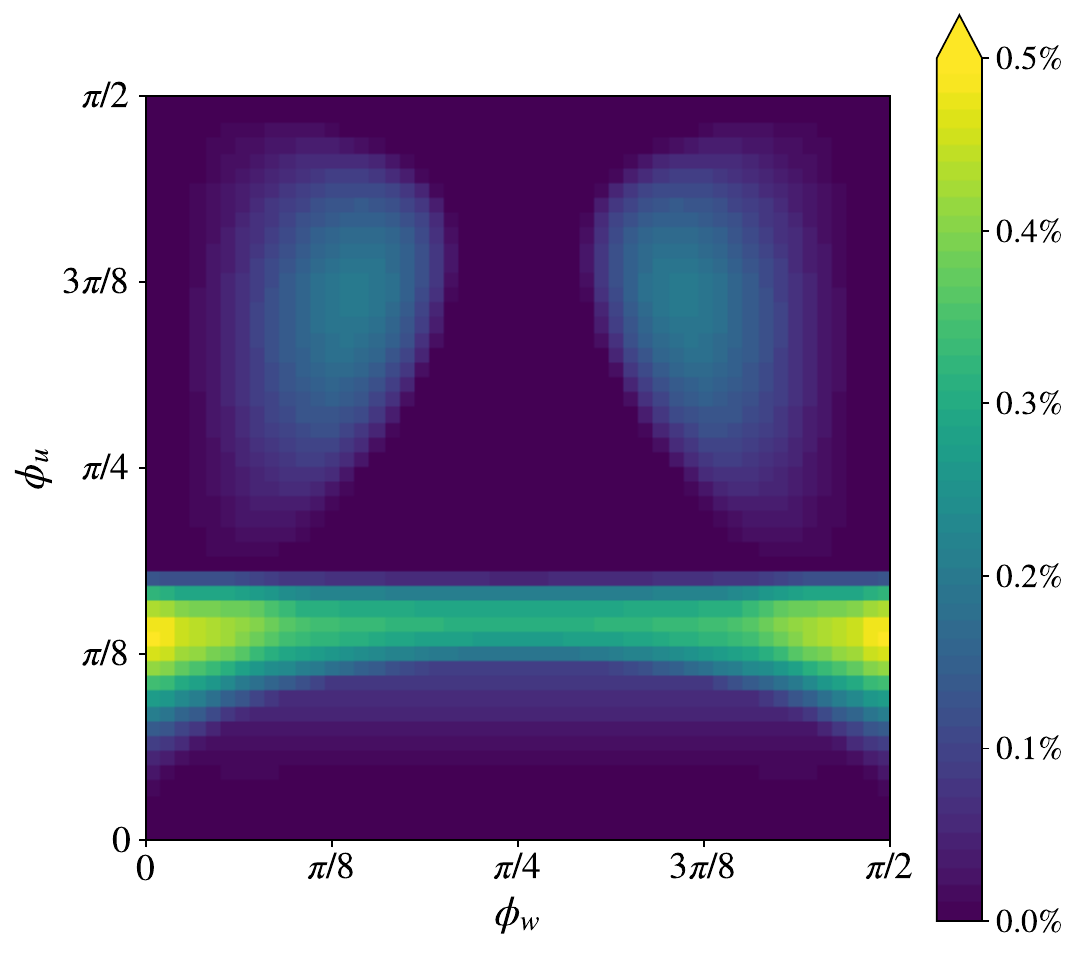}
    \caption{\textbf{White Noise Robustness.} Maximal amount of white noise below which we can prove the nonlocality of TBSM distributions. The largest amount of noise that can be tolerated is here $\omega=0.5 \%$; note that this value can be improved slightly to $\omega = 0.544 \%$ (see Appendix~\ref{app: grid}). Parameters: $u=\cos(\phi_u)$, $w=\cos(\phi_w)$ and $\lambda_0^2= 0.22$. }
    \label{fig:whitenoise}
\end{figure}

As in the ideal case to verify this feasibility problem we start by forgetting the condition $(3.iii)$ requiring the distributions  $P(\bm x_2|\bm t)$ to be local. Note also that the condition $(4.ii)$ on the total variation distance $\delta(P,\widetilde P)\leq \fe$ is equivalent to the existence of two probability distributions  $P'(\bm x_1,\bm x_2)$ and $\widetilde P'(\bm x_1,\bm x_2)$ fulfilling $P+\fe P' = \widetilde P + \fe \widetilde P'$ (see Lemma 1 in Appendix \ref{app: noisy TBSM}). By introducing these probabilities as variables we can rewire the conditions $(4.iii)$ and $(*)$ as
\begin{align}
    (4.iii')\, &\sum_{\bm x_1} \widetilde P'(\bm x_1,\bm x_2) =
   \sum_{\bm x_1}   P'(\bm x_1,\bm x_2) \\
   (*')\,\nonumber & P(\bm x_1,\bm x_2) +\fe \big(P'(\bm x_1,\bm x_2)- \widetilde P'(\bm x_1,\bm x_2)\big) \\&=\sum_{\bm t} p(\bm t) P_{\scriptscriptstyle PTC}(\bm x_1|\bm t) P(\bm x_2|\bm t)
\end{align}
and get rid of $(4.ii)$ and the variable $\widetilde P(\bm x_1,\bm x_2)$. Finally, it is also possible to directly bound the values of the token probabilities $p_\xi(t_\xi)$ from $(4.ii)$. One easily sees that the total variation distance implies $|\widetilde E_{x_1} - E_{x_1}|\leq 2 \fe$.
In fact, a slightly better bound can be obtained by using the fact that $\widetilde P$ is in the PTC slice
\be
|\widetilde E_{x_1} - E_{x_1}^\Lambda|\leq 2\fe-\ve,
\ee
where $E_{x_1}^\Lambda = \text{Pr}(x_1=0, a_1\oplus b_1\oplus c_1=1)-\text{Pr}(x_1=1, a_1\oplus b_1\oplus c_1=1)$. Hence, if all the correlators for the observed distribution satisfy $|E^\Lambda_{x_1}|> (2\fe-\ve)$ one can use Eq.~\eqref{eq: token values} to obtain
\be
\label{eq: token interval}
\begin{split}
\frac{1}{2}&\left(1+\sqrt{\frac{(|E_{b_1}^\Lambda|-(2\fe-\ve))(| E_{c_1}^\Lambda|-(2\fe-\ve))}{|E_{a_1}^\Lambda|+(2\fe-\ve)}}\right)\\
&\leq p_\alpha(0) \leq\\
 \frac{1}{2}&\left(1+\sqrt{\frac{(|E_{b_1}^\Lambda|+(2\fe-\ve))(| E_{c_1}^\Lambda|+(2\fe-\ve))}{|E_{a_1}^\Lambda|-(2\fe-\ve)}}\right),
\end{split}\ee
up to cyclic permutations.

In Appendix \ref{app: noisy TBSM} we show how using this bound the feasibility problem in Eq.~\eqref{eq: Corrolary4} can be relaxed to an LP, or a tighter collection of LPs for a grid of values $p_\xi(0)$. To illustrate this approach we now consider the family of noisy TBSM distributions, where local white noise is added to the pure state in Eq.~\eqref{eq: pure state} leading to 
the density operator
\be
\rho= (1-\omega) \ketbra{\psi}+ \omega \frac{\mathds{1}}{2}\otimes \frac{\mathds{1}}{2}.
\ee
The resulting distribution satisfies $\text{Pr}(a_1\oplus b_1 \oplus c_1=1)  = \frac{1}{2} \left(1+ (1-\omega)^3 \right)$. In Fig.~\ref{fig:whitenoise}  we plot the maximal amount of white noise $\omega$ for which we can prove the nonlocality of the distribution for $\lambda_0^2=0.22$ and different measurement parameters. The maximal value that we find is about $\omega = 0.55\%$  at $\phi_u=0.42$.

In Appendix~\ref{app: other noise} we also discuss the nonlocality of the distributions for white noise at the measurements, measurements with a no-click probability, as well as photon-number-resolving detectors with limited efficiency in the case of the single-photon entanglement implementation \cite{Abiuso2022}.

\begin{figure}
    \centering
    \includegraphics[width=\columnwidth]{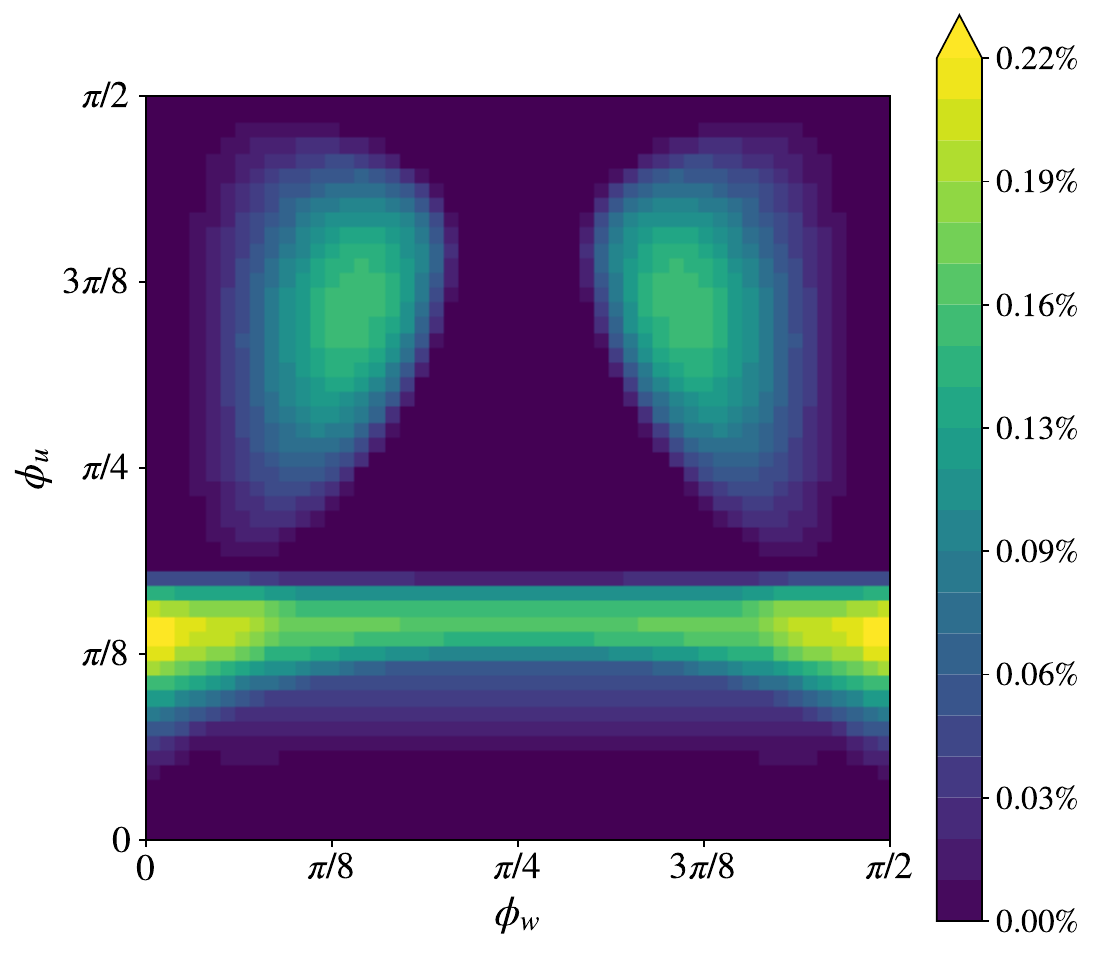}
    \caption{\textbf{Total Variation. }The size of the total variation distance ball $\mathcal{B}_\ve(\bar P)$ in Eq.~\eqref{eq: TVD ball} around the ideal TBSM distributions, for which we prove nonlocality. The maximal value in the plot is $\ve =0.22\%$, which can be improved to $\ve =0.24\%$  (see Appendix~\ref{app: grid}). Parameters: $u=\cos(\phi_u)$, $w=\cos(\phi_w)$ and $\lambda_0^2= 0.22$.}
    \label{fig:TV}
\end{figure}

\subsection{Nonlocality of regions of the quantum set.}

\label{sec: TVD}

So far we argued how the nonlocality of a given distribution $P(\bm x_1,\bm x_2)$ can be demonstrated. However, to better understand the separation between the quantum and the local sets it is more insightful to derive expressions that single out whole regions of the quantum sets that are guaranteed to be nonlocal, similarly to the usual Bell tests. In the case of nonlocality without inputs, considered here, such an expression must be nonlinear. We now derive such "Bell tests" in terms of total variation distance balls around a target quantum distribution.

Given a target quantum distribution $\bar P(\bm x_1,\bm x_2)\in \mathcal{Q}_\triangle^4$ define the total variation distance ball as the following subset of the correlation set
\be\label{eq: TVD ball}
\mathcal{B}_\ve(\bar P) = \{P\in \mathcal{P}_\triangle^4| \delta(P,\bar P) \leq \ve\}.
\ee
We will now show that for well-chosen target distributions $\bar P$ and small enough $\ve$ the whole set $\mathcal{B}_\ve(\bar P)$ is nonlocal. 

Concretely, consider $\bar P$ such that $a_1\oplus b_1 \oplus c_1=1$. Any distribution $P\in \mathcal{B}_\ve(\bar P)$ must satisfy
\be
\text{Pr}(a_1\oplus b_1 \oplus c_1=1)\geq 1-\ve.
\ee
In addition, if the ball contains some triangle-local distribution $P$ by Result 4 there must exist a local PTC distribution $\widetilde P$ with $\delta(\widetilde P(\bm x_1,\bm x_2),P(\bm x_1,\bm x_2) )\leq \fe $ by (4.\textit{ii}) and $\delta(\widetilde P(\bm x_2),P(\bm x_2) )=0$ by (4.\textit{iii}). By the triangle inequality, $\widetilde P$ must also satisfy
\begin{align}\nonumber
\delta(\widetilde P(\bm x_1,\bm x_2),\bar P(\bm x_1,\bm x_2))&\leq \delta(\widetilde P,P) +\delta(P,\bar P)\\
&\leq \fe +\ve \nonumber \\
\nonumber
\delta(\widetilde P(\bm x_2),\bar P(\bm x_2)) &\leq 0 +\delta(P,\bar P)\leq \ve,
\end{align}
leading to the following conclusion.\\

\textbf{Result 5. (Nonlocality of a ball $B_\ve(\bar P)$)}\textit{
If the ball $B_\ve(\bar P)\subset \mathcal{P}_\triangle^4$ contains triangle-local distributions, there exists a triangle-local  distribution $\widetilde P(\bm x_1,\bm x_2)$ such that}
\be\begin{split}
& (5.i) \, \widetilde P(\bm x_1,\bm x_2) = 0 \quad \text{for} \quad a_1\oplus b_1\oplus c_1\neq 1\\
&(5.ii) \,  \delta(\bar P(\bm x_1,\bm x_2),\widetilde P(\bm x_1,\bm x_2)) \leq \fe+\ve.\\
&(5.iii)\,  \delta(\bar P(\bm x_2),\widetilde P(\bm x_2)) \leq \ve.
\end{split}
\ee

Naturally, this can be combined with Result 3 in order to show that no local distribution $\widetilde P$ fulfilling $(5.i)-(5.iii)$ can exist. The argument for this is fairly similar to the corollary of Result 4.
In Appendix~\ref{app: var. dist. LP} we present an LP relaxation of this feasibility problem. In Fig.~\ref{fig:TV} we plot the size of the ball $\ve$ around the TBSM distributions with $\lambda_0^2=0.22$ which we prove to be nonlocal. The largest value we find is $\ve =0.24\%$ and it is achieved for an RGB4 distribution ($w=0$).

\section{Discussion}

\label{sec:discussion}

We have derived noise-robust proofs of quantum nonlocality in the triangle network without inputs. Our key ingredient is a result of approximate rigidity for parity token (PTC) local distributions. These methods can be applied to an arbitrary noisy distribution. In particular, we illustrate the relevance of the method by obtaining bounds on the admissible noise for a general family of quantum distributions that we introduced, generalizing the RGB4 distributions \cite{Renou_2019}. We characterized noise-robustness in terms of limited visibility for the source states, considering dephasing and white noise. Alternatively, we can also characterize distributions in the vicinity of a target distribution, via a bound on the total-variation distance.

Our work opens different perspectives. First, it would be interesting to see if our methods can be applied to larger networks and distributions satisfying PTC. This may lead to proofs of network nonlocality that are more robust to noise compared to the present results. 

Another direction is to see how to strengthen the noise robustness for the nonlocality proof in our Result 4. A promising option is to revisit Result 3, and strengthen the last condition (3.iii). We used the standard implementation of the inflation technique, but stronger versions could be obtained by exploiting nonlinear constraints. While we could prove robustness to white noise up to $0.55 \%$, there is still potential for significant improvements; indeed numerical methods~\cite{krivachy_neural_2020} suggest robustness of up to $ \sim 10 \%$. Any progress in this direction would be relevant, in particular from the perspective of experimental demonstrations. Concerning dephasing noise, we demonstrated strong robustness, up to $\sim 80 \%$. We conjecture that there exist quantum distributions that can tolerate any amount of dephasing noise arbitrarily close to one.

It would also be interesting to see if the property of approximate PTC rigidity can be extended from local models to quantum ones. If not, one can hope to find a quantum distribution violating the approximate PTC rigidity, indicating the existence of quantum nonlocal correlations within the triangle network with binary outputs and no inputs.\\

\emph{Code availability.---} Our implementation of the method for the triangle network can be found at \href{https://github.com/SadraBoreiri/Robust-Network-Nonlocality}{https://github.com/SadraBoreiri/Robust-Network-Nonlocality}\\

\emph{Acknowledgements.---} We thank Wolfgang D\"ur for discussions at the early stage of this project, and Victor Gitton for useful comments on the first version of the paper. We acknowledge financial support from the Swiss National Science Foundation (project 2000021 and NCCR SwissMAP) and by the Swiss Secretariat for Education, Research and Innovation (SERI) under contract number UeM019-3.

\bibliographystyle{quantum} 
\bibliography{main.bib}

\newpage

\onecolumngrid

\appendix

\setcounter{theorem}{0}

\section{Approximate PTC rigidity}

\label{app: approximate PTC}

In a triangle-local model, consider variables $\alpha$, $\beta$, $\gamma$ sampled from the sets $\mathit{A,B,\Gamma}$ respectively, according to some probability distributions. Without loss of generality, we can assume that the model includes deterministic response functions $a(\beta,\gamma), b(\alpha,\gamma), c(\alpha,\beta)\in\{0,1\}$\footnote{For any indeterministic local response function at the parties we can simply move the indeterminacy to the sources $\alpha$, $\beta$, $\gamma$.}.\\

Approximate PTC rigidity states that, for any triangle-local model leading to a binary output distribution $P(a,b,c)$ satisfying the approximate parity condition $\text{Pr}(a \oplus b \oplus c = 1) \geq 1- \ve$, there exists a subset $\Upsilon \subset \mathit{A}\times \mathit{B}\times \mathit{\Gamma}$ of local variables with high probability $Pr(\bm \xi \in \Upsilon)\geq 1-O(\ve)$, within which it is feasible to define the token functions $T_\alpha,T_\beta,T_\gamma$ satisfying the PTC rigidity $a(\beta,\gamma)=T_\beta\oplus T_{\gamma}\oplus 1$ etc. We prove this here, in two main steps. In the first step, we characterize the subset $\Upsilon$ and show the consistent assignment of the token functions inside this subset. In the second step, we prove that this subset happens with high probability, specifically $Pr(\Upsilon)\geq 1-3\ve$. Moreover, in Section \ref{app:Tighter_ApproxPTC} we prove a slightly more involved and tighter lower bound, showing that $Pr(\Upsilon)\geq1-\fe = 1-\frac{5}{2}\ve +O(\ve^2)$.\\

\textbf{Step 1: Defining a subset $\Upsilon$ of local variables with consistent assignment of token functions}\\

Here our objective is to specify $\Upsilon \subset \mathit{A}\times \mathit{B}\times \mathit{\Gamma}$ a subset of local variables on which we can consistently assign the token functions $T_\alpha,T_\beta,T_\gamma$ satisfying the PTC condition\\

To do so, first let $\Lambda$ be the subset of $\mathit{A}\times \mathit{B}\times \mathit{\Gamma}$ in which the outputs satisfy the PTC condition, i.e.
\begin{align}
&a(\beta,\gamma)\oplus b(\alpha,\gamma) \oplus c(\alpha,\beta)= 1\quad  \forall \,\bm \xi =(\alpha,\beta,\gamma)\in\Lambda\\
& \text{Pr}(\Lambda) = \text{Pr}(\bm \xi \in \Lambda)= 1-\ve.
\end{align}

For each value $\gamma$ define the set 
\be
G_\gamma = \{(\alpha,\beta) | (\alpha,\beta,\gamma)\in \Lambda\} \subset \mathit{A}\times \mathit{B}
\ee
as the slice of $\Lambda$ for a fixed $\gamma$. The probabilities of these sets $\text{Pr}(G_\gamma) = \text{Pr}((\alpha,\beta)\in G_\gamma)$ satisfy the following equality 
\be
\mathds{E}(\text{Pr}(G_\gamma)) = \text{Pr}((\alpha,\beta,\gamma) \in \Lambda)=\text{Pr}(\Lambda) = 1-\ve,
\ee
where the expected value on the left-hand side is taken with respect to $\gamma$. Let us now identify the value $\gamma_*$ for which $\text{Pr}(G_\gamma)$ is maximal, i.e.
\begin{align}
\gamma_* = \text{argmax}_\gamma \text{Pr}(G_\gamma)\\
1-\ve_* = \max_\gamma \text{Pr}(G_\gamma).
\end{align}
By construction we have $\text{Pr}(G_\gamma) \leq 1-\ve_*$ for all $\gamma$, and $\ve_* \leq \ve$  since the maximum is larger than the average. To shorten the notation we dub $G_*= G_{\gamma_*}$.\\

On the set $G_*$ define the token functions
\be
T_\alpha := b(\alpha,\gamma_*) \qquad T_{\beta} := a(\beta,\gamma_*) 
\ee
as well as the value $T_{\gamma_*}=1$. It is easy to see that on the set $G_*\times\{\gamma_*\}\subset \Lambda$ these functions satisfy the PTC rigidity conditions given by 
\begin{align}
T_\beta\oplus T_{\gamma_*}\oplus 1 &=  a(\beta,\gamma_*) \\
T_\alpha\oplus T_{\gamma_*}\oplus 1 &= b(\alpha,\gamma_*)\\
T_\beta \oplus T_\alpha \oplus 1 &= a(\beta,\gamma_*) \oplus b(\alpha,\gamma_*)\oplus 1 = c(\alpha,\beta),
\end{align}
where we used $a(\beta,\gamma_*) \oplus b(\alpha,\gamma_*)\oplus c(\alpha,\beta)=1$ in the last line.\\

Now consider a different value $\gamma$ and the corresponding set $G_\gamma$. Furthermore, define the set
\be
S_\gamma = G_\gamma \cap G_* = \{(\alpha,\beta)|(\alpha,\beta)\in G_*, (\alpha,\beta,\gamma)\in \Lambda \}.
\ee
It is a subset of both $G_*$ and $G_\gamma$, hence each $(\alpha,\beta)\in S_\gamma$, both $(\alpha,\beta, \gamma)$ and $(\alpha,\beta, \gamma_*)$ are inside $\Lambda$. We denote
\be
1-\delta_\gamma=\text{Pr}\big(S_\gamma),
\ee
which must satisfy $\delta_\gamma\geq \ve_*$ because $S_\gamma\subset G_*$.

Remaining in the slice $\gamma$, we can define the following binary functions 
\begin{align}
t_\alpha^{z} &:= b(\alpha,\gamma)\oplus z \oplus 1\\
t^{z}_\beta &: = a(\beta,\gamma)\oplus z \oplus 1\\
t_\gamma^z &:= z.
\end{align}
For both choices $z=0$ and $z=1$, these functions satisfy the PTC conditions on the set $S_\gamma\times\{\gamma\}$
\begin{align}
t^{z}_\alpha\oplus t_{\gamma}^z\oplus 1 &= b(\alpha,\gamma) \oplus z \oplus 1 \oplus z  \oplus 1 = b(\alpha,\gamma) \\
t^{z}_\beta\oplus t_{\gamma}^z \oplus 1 &= a(\beta,\gamma)\oplus z \oplus 1 \oplus z  \oplus 1  = a(\beta,\gamma) \\
t_\alpha^\gamma \oplus t_\beta^\gamma\oplus 1 &= b(\alpha,\gamma)\oplus a(\beta,\gamma)\oplus 1 = c(\alpha,\beta),
\end{align}
since $(\alpha,\beta,\gamma)\in \Lambda$. However, depending on the value of $z$, the functions $t^{z}_\alpha$ and $t^{z}_\beta$ might not agree with the token functions $T_\alpha$ and $T_\beta$ defined previously. More precisely, let us now define the 
\begin{align}
\mathit{A}_\gamma^z &=\{ \alpha | T_\alpha \oplus t_\alpha^z =0\}\\
\mathit{B}_\gamma^z &=\{ \alpha | T_\beta \oplus t_\beta^z =0\},
\end{align}
such that $\mathit{A}_\gamma^0$ with $\mathit{A}_\gamma^1$ and $\mathit{B}_\gamma^0$ with $\mathit{B}_\gamma^1$ are disjoint, and $S_\gamma \subset (\mathit{A}_\gamma^0 \cup \mathit{A}_\gamma^1)\times (\mathit{B}_\gamma^0 \cup \mathit{B}_\gamma^1)$. In fact for any point $(\alpha,\beta) \in S_\gamma$ it is true that 
\be
T_\alpha \oplus t_\alpha^z \oplus T_\beta \oplus t_\beta^z = (T_\alpha \oplus T_\beta)\oplus (t_\alpha^z\oplus t_\beta^z) = (c(\alpha,\beta)\oplus 1)\oplus (c(\alpha,\beta)\oplus 1) =0,
\ee
because it is both in $G_*$ and $G_\gamma$. Hence, any such point is either in $\mathit{A}_\gamma^0\times \mathit{B}_\gamma^0$ or $\mathit{A}_\gamma^1\times \mathit{B}_\gamma^1$, and we conclude that
\be
S_\gamma \subset (A_\gamma^0\times B_\gamma^0) \cup (A_\gamma^1\times B_\gamma^1).
\ee
This allows us to split $S_\gamma= S_\gamma^0\cup S_\gamma^1$ in two disjoint sets with $S_\gamma^z \subset \mathit{A}_\gamma^z\times \mathit{B}_\gamma^z$. Note that for the choice $t_\gamma=z$ the token functions $T_\alpha$ and $T_\gamma$ are PTC on $S_\gamma^z\times \{\gamma\}$. Our goal is of course to pick the value of $z$ for which $\text{Pr}(S_\gamma^z)$ is maximized. Let us define this probability 
\be
\hat P_\gamma = \max_z \text{Pr}(S_\gamma^z) \leq \Pr(S_\gamma) = 1-\delta_\gamma.
\ee
We will now lower-bound this probability. First, note that there is a trivial bound given by $\hat P_\gamma\geq \frac{1-\delta_\gamma}{2}$. Furthermore, we have
\be
\hat P_\gamma  = \text{Pr}(S_\gamma) - \min_z \text{Pr}(S^z_\gamma) \geq 1-\delta_\gamma - \min_z \text{Pr}(A^z_\gamma\times B_\gamma^z).   
\ee
To this note that for $z'= \text{argmin}_z \text{Pr}(S^z_\gamma)$ we have $\text{Pr} (S_\gamma^{z'})\leq \text{Pr} (\mathit{A}_\gamma^{z'}\times \mathit{B}_\gamma^{z'})$, and $\text{Pr} (S_\gamma^{z'})\leq \text{Pr} (S_\gamma^{z'\oplus 1}) \leq \text{Pr} (\mathit{A}_\gamma^{z'\oplus 1}\times \mathit{B}_\gamma^{z'\oplus 1})$ (since $S^z_\gamma \subset \mathit{A}^z_\gamma\times \mathit{B}_\gamma^z$).
In addition, we know that 
\be\label{app: conscosn}
1-\delta_\gamma = P_\gamma = \text{Pr} (S_\gamma^0) + \text{Pr} (S_\gamma^1)
\leq \text{Pr}(\mathit{A}_\gamma^0 \times \mathit{B}_\gamma^0) +\text{Pr}(\mathit{A}_\gamma^1 \times \mathit{B}_\gamma^1).
\ee
Denoting $p_z = \text{Pr}(\mathit{A}_\gamma^z)$ and $q_z = \text{Pr}(\mathit{A}_\gamma^z)$ with $p_0+p_1, q_0+q_1 \leq 1$, we want to upper bound $\min_z \text{Pr}(\mathit{A}^z_\gamma\times \mathit{B}_\gamma^z)$ subject to the constraint in Eq.~\eqref{app: conscosn}. This is equivalent to solving the following min-max problem
\begin{equation} \begin{split}
    \max_{p_0,p_1,q_0,q_1} &\min \, \{p_0 q_0, p_1 q_1\} \\
    \text{such that} &\quad p_0+p_1, q_0+q_1 \leq 1 \\
    & \quad p_0,p_1,q_0,q_1 \geq 0 \\
    & \quad p_0 q_0 + p_1 q_1 \geq 1-\delta_\gamma.    
\end{split}
\label{eq app: minimi}
\end{equation}
In Section \ref{app: sec minimax} we solve this optimization, proving the following bound 
\be
\min_z \text{Pr}(\mathit{A}_\gamma^z \times \mathit{B}_\gamma^z) \leq  \left(\frac{1-\sqrt{1-2 \delta_\gamma }}{2}\right)^2\quad \implies \quad \hat P_\gamma \geq 1- \delta_\gamma -\left(\frac{1-\sqrt{1-2 \delta_\gamma }}{2}\right)^2.
\ee
Recalling that $\hat P_\gamma$ must also satisfy $\hat P_\gamma \geq \frac{1-\delta_\gamma}{2}$ we obtain the desired lower bound
\be
\max_z \text{Pr}(S_\gamma^z)  \geq f(\delta_\gamma) := \max  \left\{1-\delta_\gamma - \left(\frac{1-\sqrt{1-2 \delta_\gamma }}{2}\right)^2, \frac{1-\delta_\gamma}{2} \right\}= \begin{cases}
1-\delta_\gamma - \left(\frac{1-\sqrt{1-2 \delta_\gamma }}{2}\right)^2 & \delta_\gamma \leq \frac{1}{2} \\
\frac{1-\delta_\gamma}{2} & \delta_\gamma >\frac{1}{2}
\end{cases}
\ee
Let us now define the set $\hat S_\gamma = S_\gamma^{\hat z(\gamma)}$ with $\hat z(\gamma) = \text{argmax}_z \text{Pr}(S_\gamma^z)$ to be the branch for which the probability is maximal, and prolong the token function 
\be
T_\gamma = \hat z(\gamma)
\ee
on the whole set (if $\hat S_\gamma$ is empty one can assign any value to $T_\gamma$). By taking the union of all the $\gamma$-slices we can now define the set
\be
 \Upsilon = \bigcup_\gamma \hat S_\gamma \times\{\gamma\}= \{(\alpha,\beta,\gamma) | (\alpha,\beta)\in \hat S_\gamma\}
\ee
on which we have shown a consistent assignment of the token functions $T_\alpha,T_\beta,T_\gamma$ fulfilling the PTC condition (for $\gamma=\gamma_*$ we, of course, have $\hat S_{\gamma^*} =G_*$).\\

\textbf{Step 2: proving $\text{Pr}(\Upsilon) \geq 1-3\ve$.}\\

Finally, we have to lower bound the probability  $\text{Pr}(\Upsilon)$. To do so we first consider the set
\be
\Omega = \bigcup_\gamma S_\gamma \times\{\gamma\}= \{(\alpha,\beta,\gamma) | (\alpha,\beta)\in S_\gamma \} = \{(\alpha,\beta,\gamma) | (\alpha,\beta)\in G_*,  (\alpha,\beta,\gamma)  \in \Lambda \}
\ee
for which the probability is easy to bound
\be
\text{Pr}( \Omega) =1 - \text{Pr}(\bm \xi \notin \Lambda \text{ or } (\alpha,\beta)\notin G_*)
\geq 1- \text{Pr}(\bm \xi \notin \Lambda) - \text{Pr}((\alpha,\beta)\notin G_*)
\geq 1-\ve -\ve_*.
\ee
The probabilities $\text{Pr}(\Omega)$ and $\text{Pr}(\Upsilon)$ can also be expressed as expected values over $\gamma$
\begin{align}\label{app: bound lower Upsilon}
\text{Pr}( \Omega) &=  \mathds{E}(\text{Pr}(S_\gamma)) =\mathds{E}(1-\delta_\gamma) \geq 1-\ve-\ve_*\\
\text{Pr}(  \Upsilon) &=  \mathds{E}(\text{Pr}(\hat S_\gamma))\geq  \mathds{E}(f(\delta_\gamma)).
\end{align}\\
\begin{figure}
    \centering
    \includegraphics{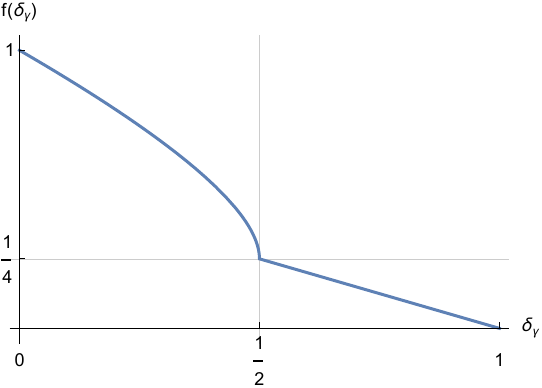}
    \caption{The function $f(\delta_\gamma)$ evaluated in the interval $[0,1]$. The concavity of the function can be observed in the $[0,1/2]$ interval.}
    \label{fig:f}
\end{figure}

The last step is to minimize the rhs of the last inequality under the constraint of Eq.~\eqref{app: bound lower Upsilon}, that is 
\begin{align}
\min \quad &\mathds{E} (f(\delta_\gamma)) \\
\text{such that} \quad & \mathds{E} (\delta_\gamma) \leq \ve+\ve_* \\
 \quad & \delta_\gamma \geq \ve_*,
\end{align}
where the last line comes from $\text{Pr}(S_\gamma) \leq \text{Pr}(G_\gamma) \leq \text{Pr}(G_*)= 1-\ve_*$. The function $f(\delta_\gamma)$ is depicted in Fig.~\ref{fig:f}. We can relax this minimization to a trivial one if we find a convex decreasing lower bound on the goal function $f(\delta_\gamma) \geq f_\text{con}(\delta_\gamma)$ on the interval $\delta_\gamma \in [\ve_*,1]$, since such a function must satisfy 
\be
\mathds{E} (f(\delta_\gamma))\geq \mathds{E} (f_\text{con}(\delta_\gamma)) \geq f_\text{con}(\mathds{E}(\delta_\gamma))  \geq  f_\text{con}(\ve+\ve_*).
\ee
Note that the function $f(\delta_\gamma)$ is concave on the interval $[0,1/2]$, which follows from the negativity of its second derivative
\be
f''(\delta_\gamma) = -\frac{1}{2 (1-2 \delta_\gamma )^{3/2}} < 0.
\ee

It follows that a convex lower bound on the whole interval $\delta_\gamma \in [0,1]$ is provided by a linear function
\be
f(\delta_\gamma)\geq f_\text{con}(\delta_\gamma)= 1-\frac{3}{2} \delta_\gamma,
\ee
which results in the following bound
\be
\text{Pr}( \Upsilon) \geq \mathds{E} (f(\delta_\gamma)) \geq 1- \frac{3}{2} (\ve+\ve_*)\geq 1-3\ve,
\ee
by $\ve_*\leq \ve$. In the next Section, we derive a tighter upper bound which uses the fact that $\delta_\gamma\geq \ve_*$. 

\subsection{Tighter bound on $\text{Pr}(\Upsilon)$}
\label{app:Tighter_ApproxPTC}

A tight convex lower bound on $f(\delta_\gamma)$ depicted in Fig.~\ref{fig:f} on the interval $[\ve_*,1]$ is given by the piecewise linear function which connects the points $\{\ve_*, f_*=f(\ve_*)\}$ with $\{1/2,1/4\}$ for $\delta_\gamma\in[\ve_*,1/2]$ and the points $\{1/2,1/4\}$ with $\{1,0\}$ for 
$\delta_\gamma\in[1/2,1]$. Formally, this lower bound is given by 
\be
f^{\ve_*}_\text{con}(\delta_\gamma)  = \begin{cases} \frac{2 f_*-\ve_*}{2-4\ve_*} - \frac{4 f_* - 1}{2-4 \ve_*} \delta_\gamma
 & \delta_\gamma \in [\ve_*,1/2] \\
\frac{1-\delta_\gamma}{2} & \delta_\gamma \in (1/2,1]
\end{cases}
\ee
This implies a tight lower bound
\begin{align}
\min \quad &\mathds{E} (f(\delta_\gamma)) \qquad  \geq f^{\ve_*}_\text{con}(\ve+\ve_*)\\
\text{such that} \quad & \mathds{E} (\delta_\gamma) \leq \ve+\ve_* \\
 \quad & \delta_\gamma \geq \ve_*,
\end{align}
and $\text{Pr}(\Upsilon) \geq f^{\ve_*}_\text{con}(\ve+\ve_*)$ for any scenario with $\text{Pr}(G_*)=1-\ve_*$. But since $\ve_*$ is not known, we need to take the worst-case lower bound
\be
\text{Pr}(\Upsilon)\geq 1-\fe := \min_{\ve_*\in[0,\ve]} f^{\ve_*}_\text{con}(\ve+\ve_*)
\ee
Finally, let us show that
\be
1-\fe= f^{\ve}_\text{con}(2\ve)=
\begin{cases}
 \frac{1}{2}-\ve +\frac{1-4\ve}{2 \sqrt{1-2 \ve }} & \ve \leq \frac{1}{4}\\
\frac{1-2\ve}{2}  & \ve >\frac{1}{4} 
\end{cases}
\ee

This follows from the fact that the function $ f^{\ve_*}_\text{con}(\ve+\ve_*)$ is decreasing in $\ve_*$ hence the minimum is reached by setting $\ve_*=\ve$. Indeed, its derivative 
\be
\frac{\partial}{\partial \ve_*} f^{\ve_*}_\text{con}(\ve+\ve_*) =
\begin{cases}
-\frac{\left(1+\sqrt{1-2 \ve_*}\right) (1-2\ve_*)-2 \ve }{2 (1-2 \ve_*)^{3/2}} & \ve+\ve_* \leq\frac{1}{2}\\
-\frac{1}{2} & \ve+\ve_* >\frac{1}{2}
\end{cases}
\ee
is negative. This concludes the proof of the lower bound
\be
\text{Pr}(\Upsilon)\geq 1-\fe \geq 1-3 \ve,
\ee
illustrated in Fig.~\ref{fig: LB}. As a final remark, consider the first order expansion $1-\fe = 1-\frac{5}{2}\ve +O(\ve^2)$ to see that for very small $\ve$ the two lower bounds are quite different.

\begin{figure}
    \centering
    \includegraphics{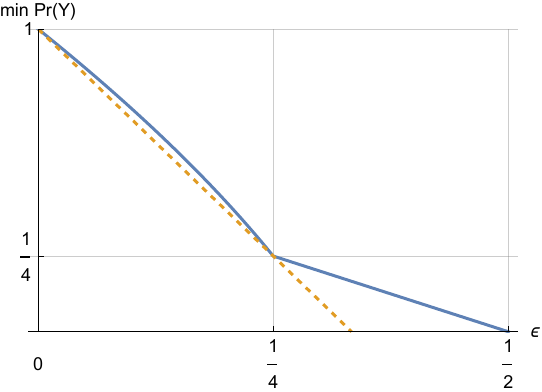}
    \caption{The lower bounds $\fe$ (full line) and $1-3 \ve$ (dashed line) on the probability  $\text{Pr}(\Upsilon)$ of the PTC set as functions of the probability $\ve= \text{Pr}(a\oplus b\oplus c \neq 1)$ to observe the wrong parity of the outcomes.}
    \label{fig: LB}
\end{figure}

\subsection{Solution of the optimization in Eq.~\eqref{eq app: minimi}}

\label{app: sec minimax}
Here we solve the optimization problem 
\begin{equation} \begin{split}
    \max_{p_0,p_1,q_0,q_1} &\min \, \{p_0 q_0, p_1 q_1\} \\
    \text{such that} &\quad p_0+p_1, q_0+q_1 \leq 1 \\
    & \quad p_0,p_1,q_0,q_1 \geq 0 \\
    & \quad p_0 q_0 + p_1 q_1 \geq 1-\delta_\gamma.    
\end{split}
\end{equation}
Without loss of generality consider the case $p_0 q_0\leq p_1 q_1$, and rewrite the problem as 
\begin{equation} \begin{split}
    \max_{p_0,p_1,q_0,q_1} & \quad p_0 q_0 \\
    \text{such that} &\quad p_0 q_0 \leq p_1 q_1
   \\ 
   & \quad p_0 q_0  \geq 1-\delta_\gamma - p_1 q_1\\
    &\quad p_0\leq 1- p_1 \\
    & \quad q_0\leq 1 - q_1 \\
    & \quad p_0,p_1,q_0,q_1 \geq 0 .\\
\end{split}
\end{equation}
The variables $p_1$ and $q_1$ now do not appear in the goal function. Furthermore, the constraints on $p_0$ and $q_0$ are easier to satisfy with larger values of $p_1$, $q_1$ and $p_1 q_1$. Hence the maximum is attained when these auxiliary variables take the maximal values $p_1=1-p_0$ and $q_1=1-p_0$. This allows us to rewrite the maximization in terms of two variables $p=p_0$ and $q=q_0$ as follows 
\begin{equation} \begin{split}
    \max_{p,q\in [0,1]} & \quad p\, q \\
    \text{such that} &\quad p\, q \leq (1-p)(1-q)
   \\ 
   & \quad p \, q  \geq 1-\delta_\gamma - (1-p)(1-q)
\end{split}
\qquad \text{which can be rewritten as}\qquad 
\begin{split}
    \max_{p,q\in [0,1]} & \quad p\, q \\
    \text{such that} &\quad p+q \leq 1 
   \\ 
   & \quad 2\, p \, q  \geq p+q-\delta_\gamma 
   \end{split}.
\end{equation}

Let us introduce the variables $g =\sqrt{p\, q}$ and $a =\frac{p+q}{2}$ which are the geometric and the algebraic means of $p,q$. With their help, we write the following optimization and they are equivalent as knowing $a,g$ gives $p,q$ as well with $p,q = a \pm \sqrt{a^2-g^2}$
\begin{equation} \begin{split}
    \max_{g,a} & \quad g^2 \\
    \text{such that} &\quad a \leq \frac{1}{2}
   \\ 
   & \quad g^2\geq a -\frac{1}{2} \delta_\gamma ,
    \\
    & \quad a \geq g.
\end{split}
\end{equation}
The final constraint arises because the algebraic and geometric means satisfy the inequality $a\geq g$, and guarantees that $p$ and $q$ are well defined, which can be saturated $a=g$  by the choice $p=q$. Note that for a fixed free variable $a$, the choice $g=a$ is the one that maximizes the goal function and relaxes the constraints the most. We can thus set $a=g$ without loss of generality in order to obtain
\begin{equation} \begin{split}
    \max_{g\in[0,\frac{1}{2}]} & \quad g^2 \\
    \text{such that} 
   & \quad g^2- g +\frac{1}{2} \delta_\gamma \geq 0.
\end{split}
\end{equation}
Here $g^2- g +\frac{1}{2} \delta_\gamma$ is a decreasing function of $g \in [0,1/2]$. It admits a zero at $g=\frac{1}{2} \left(1-\sqrt{1-2 \delta_\gamma }\right)$ inside the interval, and is thus negative for larger values of $g$. We conclude that this is precisely the value maximizing $g$ and $g^2$ under the constraint. Hence, we have shown that
\begin{equation} \begin{split}
    \max_{p_0,p_1,q_0,q_1} &\min \, \{p_0 q_0, p_1 q_1\} \qquad \quad = \left( \frac{1-\sqrt{1-2 \, \delta_\gamma }}{2}\right)^2\\
    \text{such that} &\quad p_0+p_1, q_0+q_1 \leq 1 \\
    & \quad p_0,p_1,q_0,q_1 \geq 0 \\
    & \quad p_0 q_0 + p_1 q_1 \geq 1-\delta_\gamma.    
\end{split}
\end{equation}

\section{The noisy W-distribution}
\label{app: noisy W}
In this Section, we illustrate Result 2 and discuss the nonlocality of noisy W-distributions. Such distributions $P_W^{(\ve,\delta)}$ are given by 
\begin{equation}\label{eq: W dist}
P_W^{(\ve,\delta)}(100)=\frac{1-\ve-\delta}{3} \circlearrowleft, \qquad P_W^{(\ve,\delta)}(111)=\delta
\end{equation}
and the other terms are left unspecified. As mentioned in the main text the $\circlearrowleft$ symbol means the equation is valid up to cyclic permutations of the parties. It satisfies the PTC condition with probability $1-\ve$. Hence, by our Result 2 we know that there must exist a local PTC distribution $P(a,b,c)$ such that $\delta(P_W^{(\ve,\delta)},P)\leq \fe$. With straightforward algebra, this condition can be rewritten as
\be
\Delta^{(\ve,\delta)}(P) = \sum_{(a,b,c)= (100),(010),(001)} \left|\frac{1-\ve-\delta}{3}-P(a,b,c)\right| + |\delta -P(111)|-2\fe \leq 0.
\ee
Now, the set of the PTC distributions can be parameterized with the three ``token probabilities'' $p_\xi(0)=q_\xi \in [0,1]$ leading to 
\begin{align}
    P(111) &= q_\alpha q_\beta q_\gamma + (1-q_\alpha) (1-q_\beta) (1-q_\gamma) \\
    P(100) &= (1-q_\alpha) q_\beta q_\gamma + q_\alpha (1-q_\beta) (1-q_\gamma)\qquad \circlearrowleft.
\end{align}
Therefore the implication of result 2 is equivalent to 
\be\label{app eq: minim}
\Delta_{min}^{(\ve,\delta)} = \min_{q_\alpha, q_\beta, q_\gamma} \Delta^{(\ve,\delta)}(q_\alpha q_\beta q_\gamma) \leq 0,
\ee
since there must exist a combination of $q_\alpha,q_\beta,q_\gamma$ for which $\Delta^{(\ve,\delta)}(P) \leq 0$ holds. To prove that a distribution $P_W^{(\ve,\delta)}$ is not triangle-local it is thus sufficient to show that $\Delta_{min}^{(\ve,\delta)}>0$. 

\begin{figure}
    \centering
    \includegraphics[width=0.5\columnwidth]{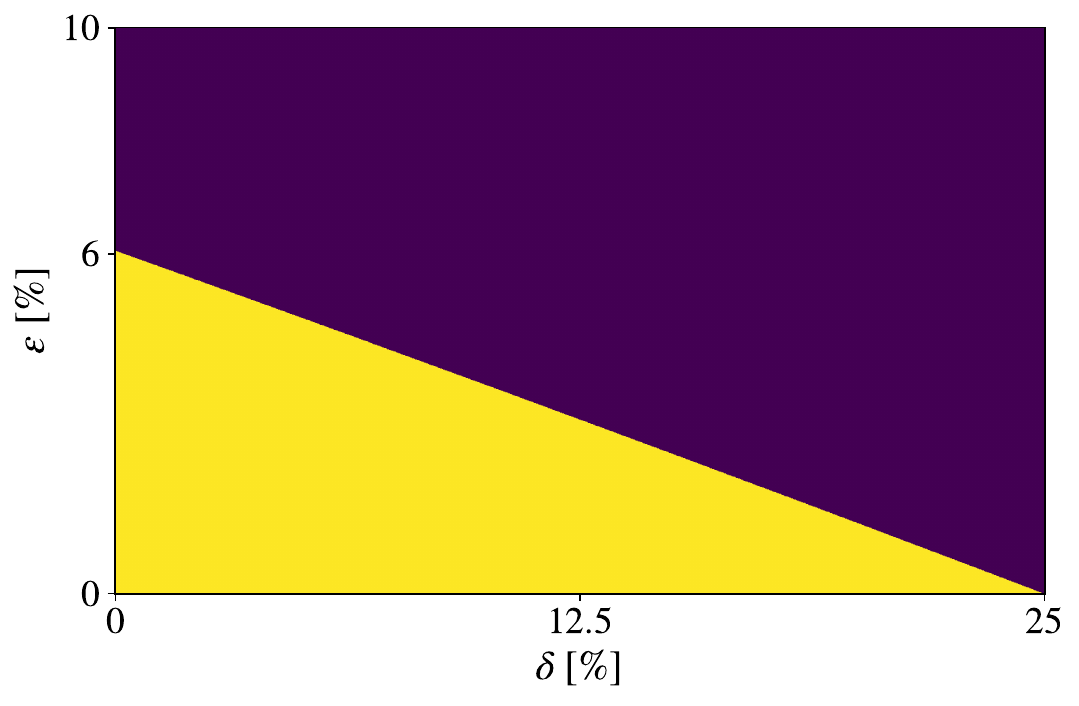}
    \caption{In yellow we show the parameter region $(\ve,\delta)$ where the noisy W-distributions $P_W^{(\ve,\delta)}(a,b,c)$ in Eq.~\eqref{eq: W dist} appear to be nonlocal.}
    \label{fig:W}
\end{figure}

To perform the minimization in Eq.~\eqref{app eq: minim}, we used the Nelder-Mead or simplex search method \cite{Nelder} via the Python library SciPy \cite{2020SciPy-NMeth}. The results of the minimization are reported in Fig.~\ref{fig:W}. Note that this is an empirical minimization algorithm, which does not provide a convergence guarantee.
Note that the goal function $\Delta^{(\ve,\delta)}(q_\alpha, q_\beta, q_\gamma)$ is a smooth and well-behaved function of three variables on a compact interval; therefore, it is plausible that the solver has found its global minimum. Nevertheless, it should be stressed that this only proves the nonlocality of a given distribution under the assumption that the global minimum was indeed found. In the future, it is desirable to perform the same minimization with more computational resources and the guarantee of convergence. This can e.g. be done with the help of the branch-and-bound algorithm, given that the goal function is Lipschitz-continuous on a compact three-parameter subset.

\section{Linear Program for TBSM when PTC is satisfied exactly} \label{app: LP excact PTC}

\begin{definition}
    For a binary tuple of outputs $\bm x = (a,b,c)$, and a binary tuple of tokens $\bm t=(t_{\alpha}, t_{\beta}, t_{\gamma})$, we define the PTC response function as follows:

    \begin{align*}
        P_{\scriptscriptstyle PTC}(\bm x|\bm t) = \begin{cases}
1 & a = t_{\beta} \oplus t_{\gamma} \oplus 1 , \ b = t_{\gamma} \oplus t_{\alpha} \oplus 1 \ , c = t_{\alpha} \oplus t_{\beta} \oplus 1\\
0 & o.w.
\end{cases}.
    \end{align*}
\end{definition}

This function is, indeed, an indicator function of whether the outputs $\bm x$ and tokens $\bm t$ are consistent. Note that clearly for any $\bm t$ there is only one consistent $\bm x$, and more interestingly for any PTC tuple $\bm x = (a,b,c)$, i.e. $a \oplus b \oplus c = 1$, there are exactly two consistent tokens $\bm t$, which are flipped of each other. For example for $\bm x =(1,1,1)$ we find $\bm t =(1,1,1)$, or $\bm t =(0,0,0)$.\\

\textbf{Noiseless case:}  Nonlocality of $P(\bm x_1 \bm x_2)$ can be checked by the infeasibility of the following linear program with
8 probability distributions (64 variables) $P(\bm x_2|\bm t)$ , where we use the compact notation $\bm x_i=(a_i,b_i,c_i)$, $\bm t =(t_\alpha, t_\beta, t_\gamma)$ 

\be\label{app:LP_noiseless}\begin{split}
& (i) \,P(\bm x_1,\bm x_2) =\sum_{\bm t} p(\bm t) P_{\scriptscriptstyle PTC}(\bm x_1|\bm t) P(\bm x_2|\bm t)\\
&(ii) \  \sum_{b_2,c_2}  {P(a_2,b_2,c_2|0,t_{\beta}, t_{\gamma})} - \sum_{b_2,c_2}P(a_2,b_2,c_2|1, t_{\beta}, t_{\gamma}) = 0 \quad \circlearrowleft
\end{split}
\ee

Note that as mentioned before, for any PTC tuple of $\bm x_1$, there are exactly two consistent $\bm t$; therefore the sum in $(i)$, is only over two terms.\\

\textbf{Dephasing noise:} 
the local dephasing noise given by the channel $\cE_d: \ketbra{\psi} \mapsto \varrho$ with
\be\label{eq: dephasing}
\varrho= (1-d) \ketbra{\psi} + d \left(\lambda_0^2 \ketbra{01} +  \lambda_1^2 \ketbra{10}
\right)\ee

where $\ket{\psi} =\lambda_0\ket{01} + \lambda_1\ket{10}$. Note that the resulting distribution $P$ in the presence of the dephasing noise is still a PTC distribution and can be written as

$$P = (1-d)^3 P_{Q} + (1-(1-d)^3) P_{C}$$

Where $P_{Q}$ is the quantum distribution obtained without the presence of dephasing noise($d=0$), and $P_{C}$ is the classical distribution($d=1$), therefore we have, $P= P_{Q} + (P_{C}-P_{Q}) e$ for $e := (1-(1-d)^3)$ and using \eqref{app:LP_noiseless} finding the minimum amount of dephasing noise $d$ for which the distribution becomes local can be written as 

\begin{align}
\text{Minimize \ } e\\
\text{s.t. \ } \ 
&(i) \ P_Q(\bm x_1 \bm x_2) = \sum_{\bm t }{ p(\bm t) P_{\scriptscriptstyle PTC}(\bm x_1|\bm t) P(\bm x_2|\bm t)} + \big(P_Q(\bm x_1 \bm x_2) - P_C(\bm x_1 \bm x_2) \big) e \\
&(ii) \  \sum_{b_2,c_2}  {P(a_2,b_2,c_2|0,t_{\beta}, t_{\gamma})} - \sum_{b_2,c_2}P(a_2,b_2,c_2|1, t_{\beta}, t_{\gamma}) = 0 \quad \circlearrowleft
\end{align}

Solving this LP, we find the optimum value $e^*$, which translates to the optimum dephasing noise $d^* = 1+\sqrt[3]{e^*-1}$. We used the MOSEK Optimizer API for Python to solve the Linear Program \cite{mosek}.

\section{Linear Program for noisy TBSM distributions}

\label{app: noisy TBSM}

\begin{lemma} \label{app:lemma1}
Let $p(x)$ and $q(x)$ be probability distributions defined on the discrete probability space $(\Omega,\mathcal{F})$. The following statements are equivalent: 

\begin{align*}
 &(i) \ \text{There exist probability distributions $p'$ and $q'$ on $(\Omega,\mathcal{F})$ such that }  p(x) = q(x) + \ve p'(x) - \ve q'(x)\\
 &(ii) \ \delta_{TV}(p,q) \leq \ve
\end{align*}

\end{lemma}

\begin{proof}
$(i)\Rightarrow (ii):$
  \begin{align*}
      \delta_{TV}(p,q) = \frac{1}{2} \sum_x |p(x) - q(x)| = \frac{\ve}{2} \sum_x |p'(x) - q'(x)| = \ve \ \delta_{TV}(p',q') \leq \ve
  \end{align*}

  $(ii)\Rightarrow (i):$
  let's define $A\subset\Omega$ as the subset within which $p(x) \geq q(x)$, then $(ii)$ implies that:
  $$ \sum_{x\in A} \big((p(x)-q(x))\big) + \sum_{x\in \bar A} \big((q(x)-p(x))\big)  \leq 2 \ve$$
  let's say $ \sum_{x\in A} \big((p(x)-q(x))\big) + \sum_{x\in \bar A} \big((q(x)-p(x))\big)  = 2 \ve'$ for $\ve' \leq \ve$. Moreover, we have $ \sum_{x\in A} \big((p(x)-q(x))\big) - \sum_{x\in \bar A} \big((q(x)-p(x))\big)  = \sum_{x\in \Omega} \big((p(x)-q(x))\big) =0$, therefore 
  \begin{equation}
      \sum_{x\in A} \big((p(x)-q(x))\big) = \sum_{x\in \bar A} \big((q(x)-p(x))\big)  =  \ve'
  \end{equation}

  Let's define $p'(x)$ and $q'(x)$ as follows:
  \begin{align*}
    p'(x)=\begin{cases}
\frac{p(x)-q(x)}{\ve}+\frac{\ve-\ve'}{\ve}f(x) & x \in A\\ 
\frac{\ve-\ve'}{\ve}f(x) & x \in \bar A
\end{cases}\\
q'(x)=\begin{cases}
\frac{\ve-\ve'}{\ve}f(x) & x \in A\\
\frac{q(x)-p(x)}{\ve}+\frac{\ve-\ve'}{\ve}f(x) & x \in \bar A
\end{cases}
  \end{align*}
  
  for any valid probability distribution $f(x)$ on the same space. It is easy to see that $p'(x)$ and $q'(x)$ are valid probability distributions satisfying (i). It is also interesting to note that, in the case of equality $\delta_{TV}(p,q) = \ve$, $p'(x)$ and $q'(x)$ respectively have a non-zero support only on $A$ and $\bar A$
\end{proof}

Note that equality $\delta_{TV}(p,q) = \ve'$ is achieved iff $\ve'$ is the maximum $\ve$ satisfying $(i)$ .\\

\textbf{LP for an observed distribution with $\text{Pr}\bm {(a_1 \oplus b_1 \oplus c_1 = 1) \geq 1- \ve}$:} Nonlocality of an observed distribution $P(\bm x_1 \bm x_2)$ can be checked by the infeasibility of the following quadratic constrained program with 6 probability distributions
$P(\bm x_2 , \bm t), P'(\bm x_1,\bm x_2), \widetilde P'(\bm x_1,\bm x_2), p_{\alpha}(t), p_{\beta}(t), p_{\gamma}(t)$.

\begin{align}
&(*')\,\nonumber  P(\bm x_1,\bm x_2) = \underbrace{\sum_{\bm t} P_{\scriptscriptstyle PTC}(\bm x_1|\bm t)  P(\bm x_2, \bm t)}_{\widetilde{P}(\bm x_1 \bm x_2)} -\fe P'(\bm x_1,\bm x_2) + \fe \widetilde P'(\bm x_1,\bm x_2)\\
&(4.iii')\, \sum_{\bm x_1} \widetilde P'(\bm x_1,\bm x_2) =
   \sum_{\bm x_1}   P'(\bm x_1,\bm x_2)\\
&(ii) \  p_{\alpha}(1) \sum_{b_2,c_2}  {P(a_2,b_2,c_2, 0,t_{\beta}, t_{\gamma})} - p_{\alpha}(0)\sum_{b_2,c_2}P(a_2,b_2,c_2, 1, t_{\beta}, t_{\gamma}) = 0 \quad \circlearrowleft
\end{align}

Note that  $P(t_\alpha,t_\beta, t_\gamma)= p_\alpha(t_\alpha)p_\beta(t_\beta)p_\gamma(t_\gamma)$ is implied by the last constraint. Furthermore, this constraint is coming from the network structure, which implies $P(a_2|0,t_{\beta}, t_{\gamma}) = P(a_2|1, t_{\beta}, t_{\gamma})$.
Moreover, note that here $p_{\alpha}(t), p_{\beta}(t), p_{\gamma}(t)$  correspond to the PTC distribution $\widetilde P$ and not the observed distribution $P$, therefore we do not know their values and we consider them as additional variables.  Note that the quadratic constraints in $(ii)$ are not convex, i.e. if we write these constraints in the form of $x^T Q x$, the matrix $Q$ is not positive semidefinite.

Note that although we do not know the distributions $p_{\alpha}(t), p_{\beta}(t), p_{\gamma}(t)$ but we can bound them:

First notice that $p_{\alpha}(t_{\alpha}=0) = \frac{1}{2} \big( 1+\sqrt{|\frac{\widetilde{E}_{b_1} \widetilde{E}_{c_1}}{\widetilde{E}_{a_1}}|}  \big) $. Where, $\widetilde{E}_{x_1} = \widetilde{P}(x_1=0)-\widetilde{P}(x_1=1)$. It is straightforward to see that the total variation distance leads to the conclusion that $|\widetilde E_{x_1} - E_{x_1}|\leq 2 \fe$.
In fact, a slightly better bound can be obtained by using the fact that $\widetilde P$ is in the PTC slice, see Section \ref{app: para 3e} below,
\be
|\widetilde E_{x_1} - E_{x_1}^\Lambda|\leq (2\fe-\ve),
\ee
where $E_{x_1}^\Lambda = P(x_1=0, a_1\oplus b_1\oplus c_1=1)-P(x_1=1, a_1\oplus b_1\oplus c_1=1)$ is computed from the PTC part of the observed distribution $P$. Hence, if all the correlators for the observed distribution satisfy $|E^\Lambda_{x_1}|\geq (2\fe-\ve)$ one can use Eq.~\eqref{eq: token values} to obtain
\be\begin{split}
l_{\alpha} := \frac{1}{2}\left(1+\sqrt{\frac{(|E_{b_1}^\Lambda|-(2\fe-\ve))(| E_{c_1}^\Lambda|-(2\fe-\ve))}{|E_{a_1}^\Lambda|+(2\fe-\ve)}}\right)&\leq   p_\alpha(0)  \\
&\leq \frac{1}{2}\left(1+\sqrt{\frac{(|E_{b_1}^\Lambda|+(2\fe-\ve))(| E_{c_1}^\Lambda|+(2\fe-\ve))}{|E_{a_1}^\Lambda|-(2\fe-\ve)}}\right) := u_{\alpha} \quad \circlearrowleft .
\end{split}\ee

Therefore $(ii)$ can be written as:
\begin{align}\label{eq: fvabfl}
\  l_{\alpha} \sum_{b_2,c_2}  {P(a_2,b_2,c_2, 0,t_{\beta}, t_{\gamma})} - (1- l_{\alpha}) \sum_{b_2,c_2}P(a_2,b_2,c_2, 1, t_{\beta}, t_{\gamma}) \leq 0 \\\label{eq: fvabf}
\  u_{\alpha} \sum_{b_2,c_2}  {P(a_2,b_2,c_2, 0,t_{\beta}, t_{\gamma})} - (1- u_{\alpha}) \sum_{b_2,c_2}P(a_2,b_2,c_2, 1, t_{\beta}, t_{\gamma}) \geq 0.
\end{align}

As an illustration, we can obtain bounds on the white noise robustness of the RGB4 distributions, see Fig.~\ref{fig:RGB4_whitenoise}.

\begin{figure}[h!]
    \centering    \includegraphics[width=0.5\columnwidth]{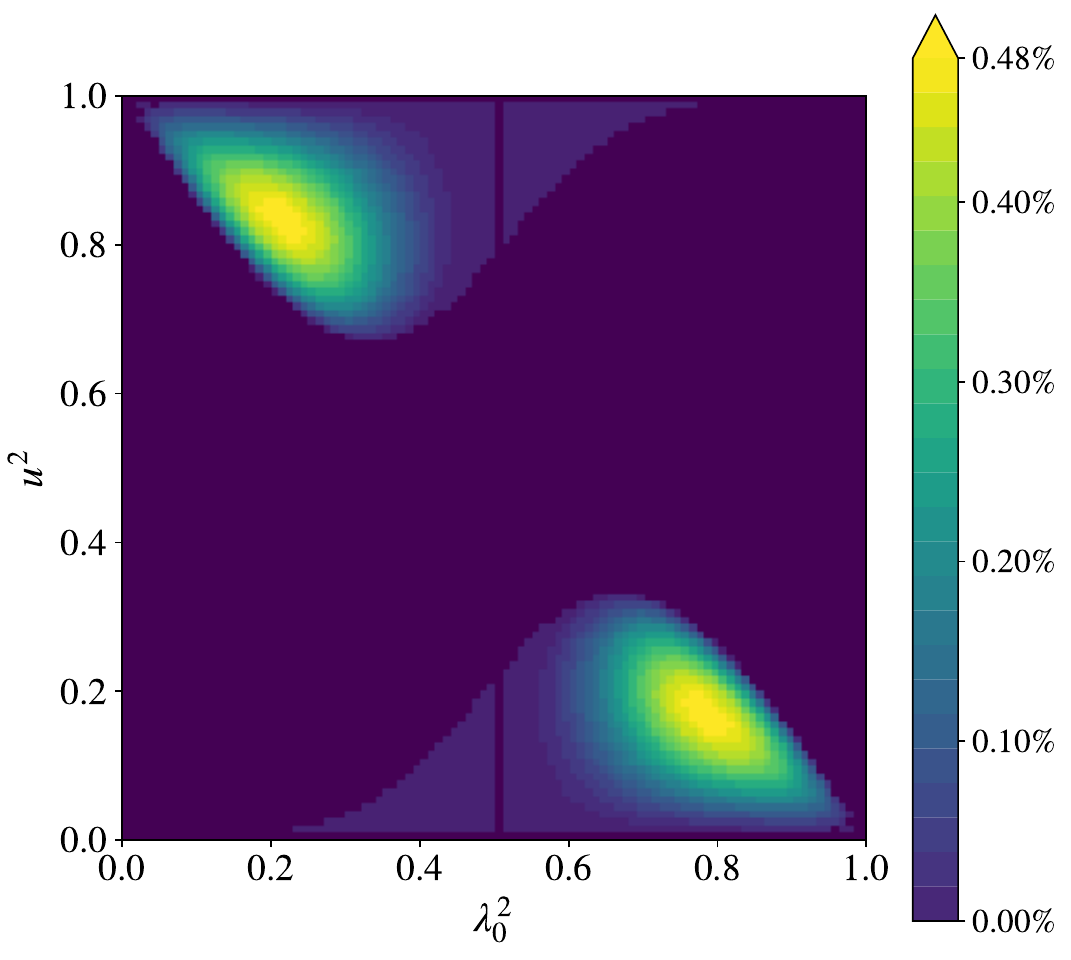}
    \caption{\textbf{Noise robustness of the RGB4 distributions.} Maximal amount of white noise below which we can prove the nonlocality of RGB4 distributions (TBSM distributions with $\phi_w=0$).}
    \label{fig:RGB4_whitenoise}
\end{figure}

\subsection{Collection of LPs (the grid technique for $p_\xi(t_\xi)$)}
\label{app: grid}
Finally, note that the LP has to be satisfied for some values $p_\xi(t_\xi)$ of the token probabilities in the intervals $[l_\alpha,u_\alpha]$. In practice, it is convenient to divide each interval into smaller intervals $[l_\alpha +k\frac{u_\alpha-l_\alpha}{M}, l_\alpha +(k+1)\frac{u_\alpha-l_\alpha}{M}]$, and verify that no solution can be found within any of the $M^3$ possible combinations of intervals. This gives a better result as compared to the harsher relaxation in Eqs.~(\ref{eq: fvabfl},\ref{eq: fvabf}). In particular, for the figures \ref{fig:whitenoise},\ref{fig:TV},\ref{fig:RGB4_whitenoise},\ref{fig:white_noise_SideBySide}, and \ref{fig:NoClick_SinglePhotone} we used $M=16$, while to obtain the best bounds $\omega=0.544\%, \omega'=0.504\%$ white noise and $\ve=0.24\%$ total-variation distance we went to $M=128$.\\

Furthermore, to see how good is this relaxation of the $p_\xi(t_\xi)$ variable to a finite set of intervals, we also check the feasibility of the LP for the honest values of $p_\xi(0)=q_\xi = \lambda_0^2$, i.e. the values which appear in the quantum strategy and are thus feasible by construction. This gives an upper bound on the maximum amount of noise for which the original feasibility problem could guarantee that the distribution is nonlocal.

\subsection{Development of a tighter bound on $\widetilde {E}_{x_1}$} 
\label{app: para 3e}
Note that we have:
\begin{align*}
P(\bm x_1, \bm \xi \in \Lambda) = P(\bm x_1; \bm \xi \in \Upsilon) + (\fe-\ve) P''(\bm x_1)\\
\widetilde{P}(\bm x_1) = P(\bm x_1; \bm \xi \in \Upsilon) + \fe \widetilde{P'}(\bm x_1)
\end{align*}
Which can be rewritten as :
\begin{align*}
\widetilde{P}(\bm x_1) = P(\bm x_1, \bm \xi \in \Lambda) + \fe \widetilde P'(\bm x_1) - (\fe-\ve) P''(\bm x_1)
\end{align*}

Therefore with $E_{x_1}^\Lambda = P(x_1=0, a_1\oplus b_1\oplus c_1=1)-P(x_1=1, a_1\oplus b_1\oplus c_1=1)$ we have

\begin{align*}
\widetilde{E}_{x_1} = E_{x_1}^\Lambda + \fe \widetilde E'_{x_1} - (\fe-\ve) E''_{x_1}
\end{align*}
Therefore,
\begin{align*}
\widetilde {E}_{x_1} \in [ E_{x_1}^\Lambda - (2\fe- \ve), E_{x_1}^\Lambda + (2\fe-\ve)].
\end{align*}

\newpage
\section{Other noise models}
\label{app: other noise}

{\paragraph{\textbf{White noise at the measurement.}}
Alternatively to the white noise at the source, discussed in the main text, one can consider white noise affecting the  measurements. In this case the POVM elements describing the measurement are replaced by 
\be
E_{x_1,x_2}=(1-\omega' ) \ketbra{\phi_{x_1,x_2}} + \omega ' \frac{\mathds{1}}{2}\otimes \frac{\mathds{1}}{2}.
\ee
Note that here for the  resulting distribution  we have $\text{Pr}(a_1\oplus b_1 \oplus c_1=1)  = \frac{1}{2} \left(1+ (1-\omega ')^3 \right)$. In Fig.~ \ref{fig:white_noise_SideBySide}(a),  we depict the robustness to measurement white noise of the TBSM distributions with $\lambda_0^2=0.22$. The maximal  amount of noise for which we can prove that nonlocality of the distribution is given by $\omega'=0.504\%$ for the case $\phi_u=0.42$ and $\phi_w=0$.

Considering the combined effect of white noise at the sources and the parties we get, 
$\text{Pr}(a_1\oplus b_1 \oplus c_1=1)  = \frac{1}{2} \left(1+ (1-\omega ' )^3(1-\omega)^3 \right)$. In Fig. \ref{fig:white_noise_SideBySide}(b),
we depict the nonlocality of the TBSM distribution with $\phi_u=0.42$ and $\phi_w=0$ as a function of both the noise parameters.  
\begin{figure*}[h!]
    \centering
    \subfloat[\centering \textbf{White noise at the measurement}]{{\includegraphics[width=0.52\textwidth]{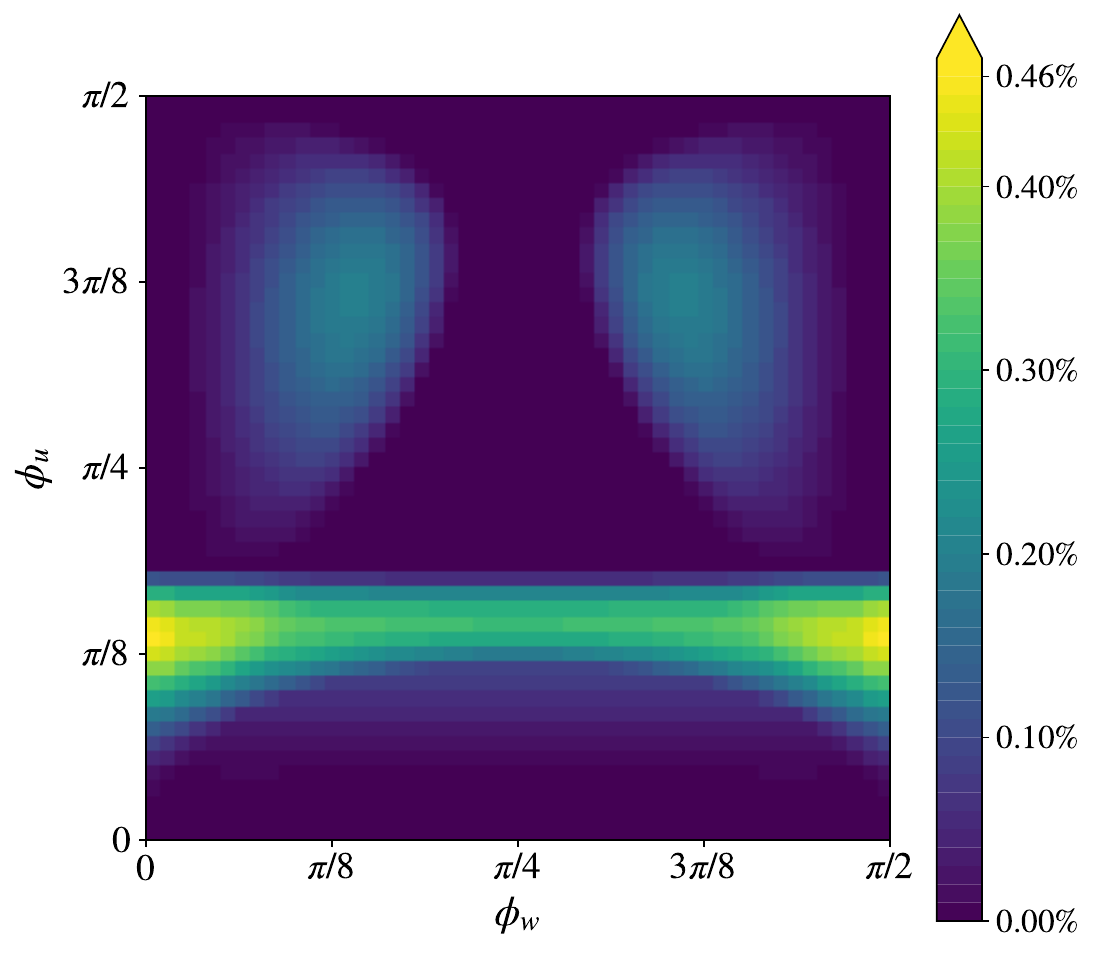} }}%
    \qquad
    \subfloat[\centering \textbf{Noisy states and measurements}]{{\includegraphics[width=0.41\textwidth]{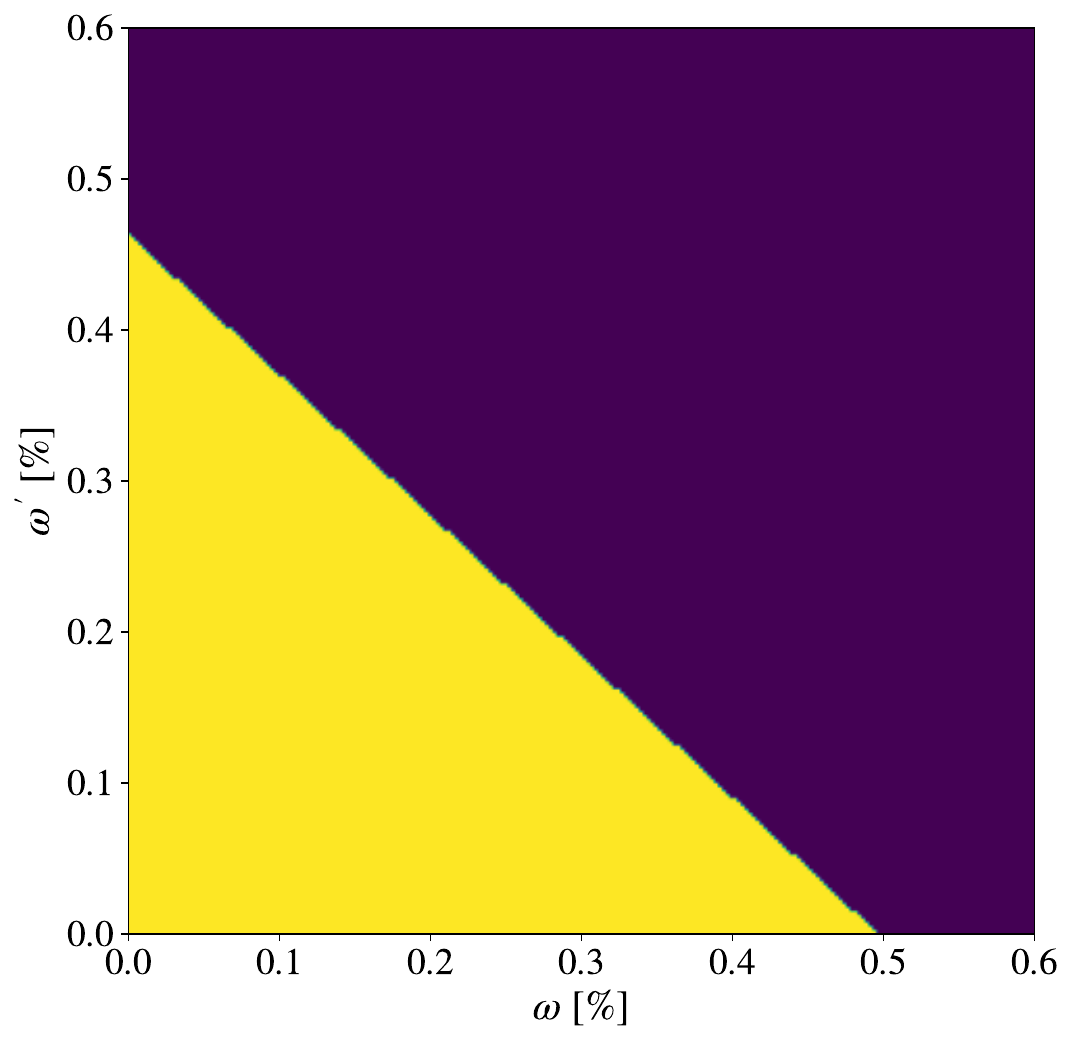} }}%
    \caption{ \textbf{(a)} Maximal amount of white noise at the level of measurements per party below which we can prove the nonlocality of the TBSM distribution.  Parameters: $u=\cos(\phi_u)$, $w=\cos(\phi_w)$ and $\lambda_0^2= 0.22$. \textbf{(b)} Noise robustness of the TBSM distribution when noise is added to both states ($\omega$) and measurements ($\omega'$). The distribution is nonlocal in the yellow region, for parameters $\lambda_0^2 = 0.22, \phi_u=0.42, \phi_w=0$.} \label{fig:white_noise_SideBySide}
\end{figure*}

\vspace{0.5cm}

\paragraph{\textbf{No-click noise.}} Another natural imperfection to consider is the possibility that the measurement does not produce an outcome with probability $p$. This is described by five-outcome POVM given by
\begin{align}
E_{x_1,x_2} &= (1-p) \ketbra{\phi_{x_1,x_2}} \qquad \forall x_1,x_2=0,1 \\
E_\nc &= p\,  \id
\end{align}
where $E_\nc $ is the element associated with the no-click event. The probability distribution obtained with such measurements has five outcomes for each party and is not readily analyzed with the results presented in this paper. However, a trick commonly used for device-independent quantum key distribution is to relabel the no-click outcome to one of the four "good" outcomes restoring output cardinality.

We  find that the optimal strategy (for noise robustness) is for each party to relabel the no-click outcome to $\nc \to (1,0)$, which corresponds to the POVM
\begin{align}
E_{1,0} &= (1-p) \ketbra{\phi_{1,0}} + p \, \id \\
E_{x_1,x_2} &= (1-p) \ketbra{\phi_{x_1,x_2}}  \qquad \forall (x_1,x_2)\neq (1,0).
\end{align}
For this strategy, we prove the nonlocality of the resulting distribution for up to a no-click probability of $p = 0.7\%$, see Fig.~ \ref{fig:NoClick_SinglePhotone}(a).\\

Note that if e.g. the sources produce two photons entangled in polarization $\ket{0}=\ket{1_h}$ and $\ket{1}=\ket{1_v}$, a limited photon transmission (or detector efficiency) $\eta$ would correspond to the no-click noise with $(1-p) =\eta^2$.

\vspace{0.5cm}

\paragraph{\textbf{Single photon entanglement implementation with loss.}}

For $\phi_w=0$ the TBSM distributions (RGB4) can be realised with single-photon entanglement (vacuum state $\ket{0}$ and single photon $\ket{1}$), linear optics and photon number resolving detectors~\cite{Abiuso2022}. In this case the dominant noise is the limited efficiency of the detectors, which can be modeled as a photon loss channel with transmission $\eta$. Formally, we have the channel
\begin{align}
\cE(\rho) &= K_0 \rho K_0^\dag+ K_1 \rho K_1^\dag\\
K_0 &=\ketbra{0} + \sqrt{\eta} \ketbra{1}{1} \\
K_1 &= \sqrt{1-\eta} \ketbra{0}{1}
\end{align}
which is applied on all of the six modes. Alternatively, the noise can be absorbed in the measurements by considering the POVM elements 
\be
E_{x_1,x_2}= \cE^*\otimes \cE^*(\Pi_{x_1,x_2})=\sum_{i,j=0}^1 K_i^\dag\otimes K_j^\dag \, \Pi_{x_1,x_2}\,  K_i\otimes K_j.
\ee
Recall that for $\phi_w=0$ the original projectors for $x_1=1$ read $\Pi_{00}=\ketbra{00}$ and $\Pi_{01}=\ketbra{11}$, while the lossy POVM elements read
\begin{align}
E_{0,0} &= \left(
\begin{array}{cccc}
 1 & 0 & 0 & 0 \\
 0 & 1-\eta  & 0 & 0 \\
 0 & 0 & 1-\eta  & 0 \\
 0 & 0 & 0 & (1-\eta )^2 \\
\end{array}
\right) \qquad\qquad \quad E_{0,1} = \left(
\begin{array}{cccc}
 0 & 0 & 0 & 0 \\
 0 & 0 & 0 & 0 \\
 0 & 0 & 0 & 0 \\
 0 & 0 & 0 & \eta ^2 \\
\end{array}
\right)\\
E_{1,0}&=\left(
\begin{array}{cccc}
 0 & 0 & 0 & 0 \\
 0 & \eta  u^2 & \eta  u v & 0 \\
 0 & \eta  u v & \eta  v^2 & 0 \\
 0 & 0 & 0 & (1-\eta ) \eta  \\
\end{array}
\right) \qquad 
E_{1,1}=\left(
\begin{array}{cccc}
 0 & 0 & 0 & 0 \\
 0 & \eta  v^2 & -\eta  u v & 0 \\
 0 & - \eta  u v & \eta  u^2 & 0 \\
 0 & 0 & 0 & (1-\eta ) \eta  \\
\end{array}
\right)
\end{align}
in the computational basis $\{\ket{00},\ket{01},\ket{10},\ket{11}\}$.\\ 

The nonlocality of the resulting distributions can be detected up to a transmission loss of $1-\eta=0.275 \%$, see Fig.~\ref{fig:NoClick_SinglePhotone} (b).

\begin{figure*}[htbp]
    \centering
    \subfloat[\textbf{No-click scenario}]{%
        \includegraphics[width=0.52\textwidth]{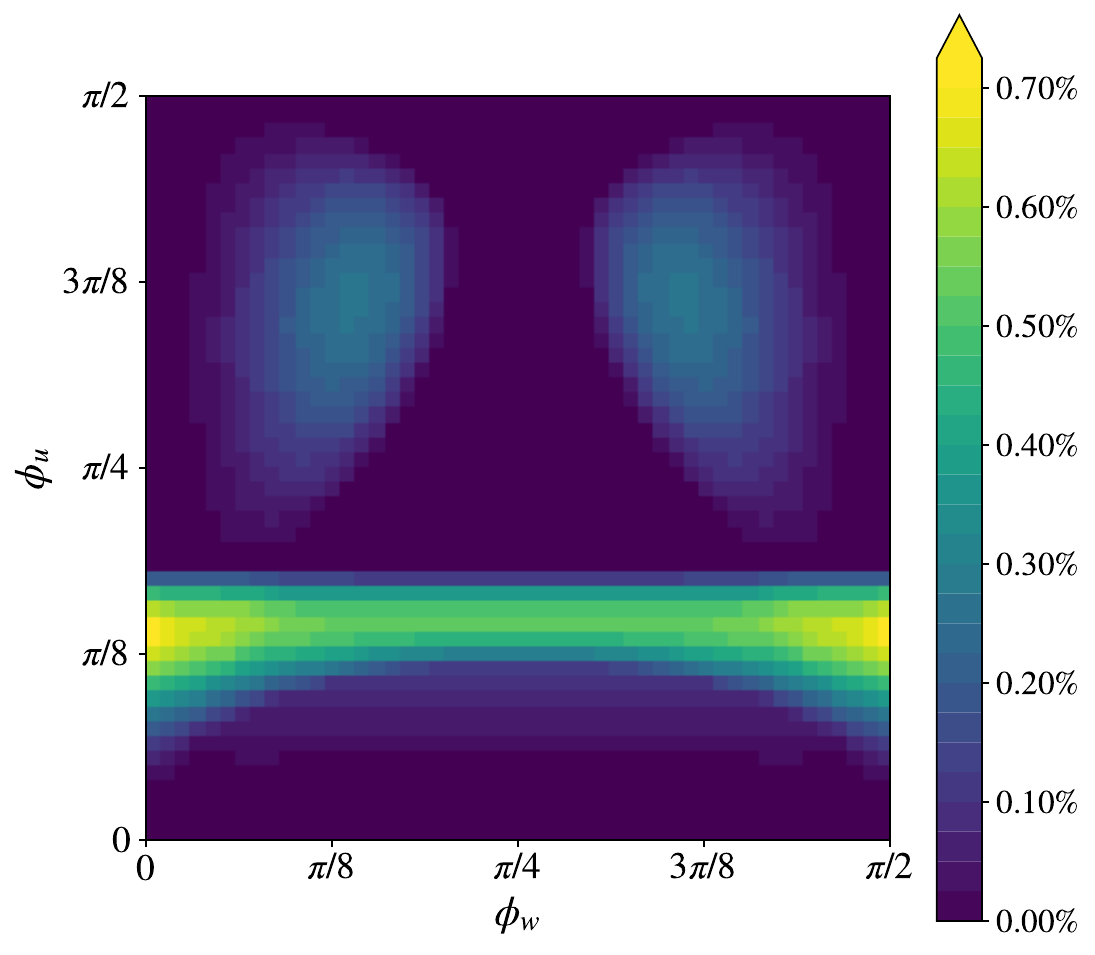}
    }
    \hfill
    \subfloat[\textbf{Single-photon entanglement scenario with loss}]{%
        \includegraphics[width=0.46\textwidth]{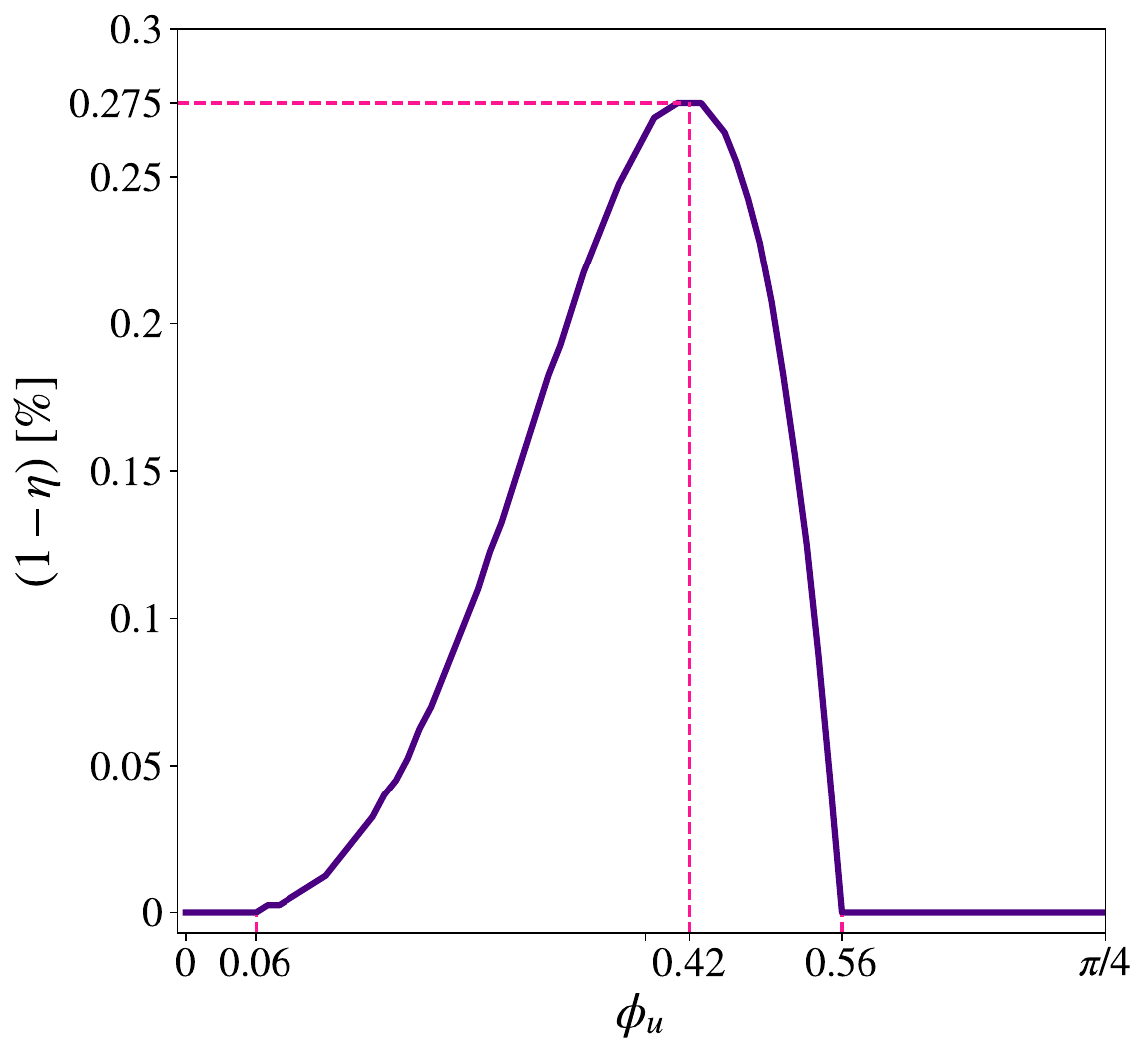}
    }
    \caption{ \textbf{(a)} Maximal no-click probability at each party below which we can prove the nonlocality of the distribution. Parameters: $u=\cos(\phi_u)$, $w=\cos(\phi_w)$ and $\lambda_0^2= 0.22$. \textbf{(b)} Maximal amount of photon loss at measurements below which we can prove the nonlocality of the RGB4 distribution (TBSM with $\phi_w=0$) for the single-photon entanglement implementation. Here $\eta$ denotes the detection efficiency, i.e. the probability that an emitted photon is detected. Parameters: $u=\cos(\phi_u)$, $\phi_w=0$ and $\lambda_0^2= 0.22$}
    \label{fig:NoClick_SinglePhotone}
\end{figure*}

\newpage

\section{Linear Program for the total variation distance}
\label{app: var. dist. LP}

Given a target quantum distribution $\bar P(\bm x_1,\bm x_2)\in \mathcal{Q}_\triangle^4$ define the total variation distance ball as the following subset of the correlation set
\be\label{eq: TVD ball2}
\mathcal{B}_\ve(\bar P) = \{P\in \mathcal{P}_\triangle^4| \delta(P,\bar P) \leq \ve\}.
\ee
We will now demonstrate that for well-chosen target distributions $\bar P$ and small enough $\ve$ the entire set $\mathcal{B}_\ve(\bar P)$ is nonlocal. Concretely, consider $\bar P$ such that $a_1\oplus b_1 \oplus c_1=1$. Any distribution $P\in \mathcal{B}_\ve(\bar P)$ must satisfy
\be
\text{Pr}(a_1\oplus b_1 \oplus c_1=1)\geq 1-\ve.
\ee

We will prove that all distributions in $\mathcal{B}_\ve(\bar P)$ are nonlocal by contradiction.  Let us assume that the ball contains some triangle-local distribution $P$, then from Result 3 there must exist a local PTC distribution $\widetilde P$ with $\delta(\widetilde P(\bm x_1,\bm x_2),P(\bm x_1,\bm x_2) )\leq \fe$ and $\delta(\widetilde P(\bm x_2),P(\bm x_2) )=0$. By triangle inequality, $\widetilde P$ must also satisfy
\begin{align}
\delta(\widetilde P(\bm x_1,\bm x_2),\bar P(\bm x_1,\bm x_2))&\leq \delta(\widetilde P,P) +\delta(P,\bar P)\leq \ve + \fe \label{app:eq:TV_3eps_x12}\\
\delta(\widetilde P(\bm x_2),\bar P(\bm x_2)) &\leq 0 +\delta(P,\bar P)\leq \ve.\label{app:eq:TV_eps_x2}
\end{align}

Therefore, from Eq. \eqref{app:eq:TV_3eps_x12} and Lemma \ref{app:lemma1}

\begin{equation} \label{app:eq:TV_1}
 \bar P(\bm x_1,\bm x_2) = \underbrace{\sum_{\bm t} P_{\scriptscriptstyle PTC}(\bm x_1|\bm t)  P(\bm x_2, \bm t)}_{\widetilde{P}(\bm x_1 \bm x_2)} -(\ve+\fe) \bar P'(\bm x_1,\bm x_2) + (\ve+\fe) \widetilde P'(\bm x_1,\bm x_2).
\end{equation}

Taking the marginal we get

\begin{equation} 
 \bar P(\bm x_2) = \widetilde{P}(\bm x_2) -(\ve+\fe) \bar P'(\bm x_2) + (\ve+\fe) \widetilde P'(\bm x_2).
\end{equation}

Now from this equation and  Eq. \eqref{app:eq:TV_eps_x2} we get

\begin{equation}
    \delta(\widetilde P'(\bm x_2),\bar P'(\bm x_2)) = \frac{1}{\ve + \fe} \delta(\widetilde P(\bm x_2),\bar P(\bm x_2)) \leq \frac{\ve}{\ve + \fe} := \alpha.
\end{equation}

Note that as $2.5 \ve \leq \fe \leq 3 \ve$, we have $ 1/4 \leq \alpha \leq 2/7$, moreover for small values of $\ve$, we have $\fe \approx 2.5 \ve$ and $\alpha \approx 2/7$. Again by Lemma \ref{app:lemma1}, we have

\begin{align} \label{app:eq:TV_2}
(4.iii')\, \sum_{\bm x_1} \widetilde P'(\bm x_1,\bm x_2) -
   \sum_{\bm x_1}   \bar P'(\bm x_1,\bm x_2) + \alpha \widetilde P''(\bm x_2) - \alpha \bar P''(\bm x_2) = 0.
\end{align}

Moreover, like before, we use the network constraint that $P(a_2|0,t_{\beta}, t_{\gamma}) = P(a_2|1, t_{\beta}, t_{\gamma})$ written as the following

\begin{align}
(ii) \  p_{\alpha}(0) \sum_{b_2,c_2}  {P(a_2,b_2,c_2, 0,t_{\beta}, t_{\gamma})} - p_{\alpha}(1)\sum_{b_2,c_2}P(a_2,b_2,c_2, 1, t_{\beta}, t_{\gamma}) = 0 \quad \circlearrowleft.
\end{align}

Note that here $p_{\alpha}(t), p_{\beta}(t), p_{\gamma}(t)$  correspond to the PTC distribution $\widetilde P$ which is not observed, but we can bound $p_{\alpha}(t), p_{\beta}(t), p_{\gamma}(t)$:

Notice that $p_{\alpha}(t_{\alpha}=0) = \frac{1}{2} \big( 1+\sqrt{|\frac{\widetilde{E}_{b_1} \widetilde{E}_{c_1}}{\widetilde{E}_{a_1}}|}  \big) $. Where, $\widetilde{E}_{x_1} = \widetilde{P}(x_1=0)-\widetilde{P}(x_1=1)$. It is straightforward to see that the total variation distance $\delta (\widetilde P, \bar P) \leq \fe+\ve$, leads to the conclusion that $|\widetilde E_{x_1} - \bar E_{x_1}|\leq 2(\fe+\ve) $. Where $\bar E_{x_1} = \bar{P}(x_1=0)-\bar{P}(x_1=1)$ is computed from the target PTC distribution $\bar P$. Hence, if all the correlators for the target distribution satisfy $|\bar E_{x_1}|\geq 2(\fe+\ve)$ one can use Eq.~\eqref{eq: token values} to obtain
\begin{align}\begin{split}
l_{\alpha} := \frac{1}{2}\left(1+\sqrt{\frac{(|\bar E_{b_1}|-2(\fe+\ve))(| \bar E_{c_1}|- 2(\fe+\ve))}{|\bar E_{a_1}|+ 2(\fe+\ve)}}\right)\leq \  p_\alpha(0) \ \\ p_\alpha(0) \leq \frac{1}{2}\left(1+\sqrt{\frac{(|\bar E_{b_1}|+ 2(\fe+\ve))(|\bar  E_{c_1}|+ 2(\fe+\ve))}{|\bar E_{a_1}|- 2(\fe+\ve)}}\right) := u_{\alpha} \quad \circlearrowleft .
\end{split}\end{align}

Therefore $(ii)$ can be written as:
\be
\begin{split}
\  l_{\alpha} \sum_{b_2,c_2}  {P(a_2,b_2,c_2, 0,t_{\beta}, t_{\gamma})} - (1- l_{\alpha}) \sum_{b_2,c_2}P(a_2,b_2,c_2, 1, t_{\beta}, t_{\gamma}) \leq 0 \\
\  u_{\alpha} \sum_{b_2,c_2}  {P(a_2,b_2,c_2, 0,t_{\beta}, t_{\gamma})} - (1- u_{\alpha}) \sum_{b_2,c_2}P(a_2,b_2,c_2, 1, t_{\beta}, t_{\gamma}) \geq 0.
\end{split} \label{app:eq:TV_ineq}
\ee

Combining Equations \ref{app:eq:TV_1}, \ref{app:eq:TV_2}, and \ref{app:eq:TV_ineq}, yields a set of linear constraints that should be satisfied if there exists any local distribution inside the ball $\mathcal{B}_\ve(\bar P)$. Therefore, the infeasibility of this LP proves the non-locality of all of the points in $\mathcal{B}_\ve(\bar P)$. (see Fig.~ \ref{fig:TV}). Again, to find tighter bounds, we use the grid technique.

\end{document}